\documentclass{article}

\usepackage{PRIMEarxiv}
\usepackage[utf8]{inputenc} 
\usepackage[T1]{fontenc}    
\usepackage{hyperref}       
\usepackage{cite}
\usepackage{url}            
\usepackage{booktabs}       
\usepackage{amsfonts}       
\usepackage{nicefrac}       
\usepackage{microtype}      
\usepackage{lipsum}
\usepackage{fancyhdr}       
\usepackage{graphicx}       
\graphicspath{{media/}}     
\usepackage[normalem]{ulem}
\usepackage{array}
\usepackage{amssymb}
\usepackage{tikz}
\usetikzlibrary{positioning,quotes}
\usepackage{subcaption}     
\usepackage{xcolor}
\usepackage{soul}
\usepackage{amsmath}
\usepackage{amsthm}
\usepackage{framed}
\usepackage{mathtools}

\usepackage{graphicx}
\usepackage{subcaption}
\usepackage{booktabs}

\usepackage{algorithm}
\usepackage{algpseudocode}

\theoremstyle{plain}
\newtheorem{theorem}{Theorem}        
\newtheorem{proposition}{Proposition}
\theoremstyle{definition}
\newtheorem{remark}{Remark}

\theoremstyle{definition}
\newtheorem{definition}{Definition}  

\newcommand{\ket}[1]{{| #1 \rangle}}
\newcommand{\bra}[1]{\langle #1 |}

\usepackage{tikz}

\usepackage[utf8]{inputenc}
\usepackage{authblk}
\usepackage{hyperref}
\usepackage{setspace} 
\title{Heisenberg-Limited Quantum Hamiltonian Learning \\ via Randomly Spread Product-States}

\author[1,2]{Bora Baran\thanks{\texttt{b.baran@fz-juelich.de}\;}\;\;}

\author[1,3]{Timothy Heightman\thanks{\texttt{timothy.heightman@icfo.eu}\;}\;\;}

\affil[1]{\textit{ICFO - The Institute of Photonic Sciences, 08860 Castelldefels, Barcelona, Spain}}

\affil[2]{Forschungszentrum Jülich, Institute of Quantum Control, Peter Grünberg Institut (PGI-8), 52425 Jülich, Germany}
\affil[3]{Quside Technologies SL, Carrer d’Esteve Terradas, 1, 08860 Castelldefels, Barcelona, Spain}

\begin{document}

\maketitle

\begin{abstract}
We show how the Heisenberg-limited quadratic Fisher-information regime of short-time quantum evolution can be made practically accessible for quantum Hamiltonian learning, using only local quantum operations. Our protocol uses experiments initialized in locally Haar-random product states, accompanied by random one-shot Pauli-product measurements, leading to the activation of the full Hamiltonian spectrum in the measurement statistics. This extends the naturally given quadratic Fisher scaling of short-time dynamics into a practically accessible temporal window without requiring entanglement, globally coherent measurements, or dynamical control. Furthermore, we show that the act of ensemble averaging over these initial states makes unbiased estimation data, meaning all Hamiltonian parameters can be simultaneously estimated from the same data-set, removing the need for parameter isolation. We supplement the theoretical results by showing empirically that, even away from the asymptotic limit, one can surpass the SQL using randomly spread product-state ensembles. We do so numerically by learning a selection of different disordered multi-qubit Hamiltonians in a black-box learning scenario.

\end{abstract}



\section*{Introduction}
\label{sec:introduction}

Hamiltonian Learning (HL) is a form of quantum process tomography, and is a crucial subroutine in quantum information processing. It aims to characterise an unknown system's unitary dynamics, with applications in quantum simulation, control, sensing, metrology and error mitigation in quantum computing \cite{zhang2017quantum, hauke2012can, Shi300qubits, WeibeHLCert, valenti2019hamiltonian, guo2025mitigating, baumgratz2016quantum, ferrie2013best, singh2023compensating, sergeevich2011characterization, endo2019mitigating, fallani2022learning, strikis2021learning, suzuki2022quantum, yang2020probe, kobrin2024universal, imai2025metrological, peng2022deep}.  Characterizing a quantum mechanical system's dynamics typically involves learning the parameters of the system's Hamiltonian from experimental data \cite{bairey2019learning, evans2019scalable, franca2022efficient, wilde2022scalably, gu2022practical, heightman2024solving, dutkiewicz2023advantage, hu2025ansatz, ma2024learning, barthe2025quantum, leng2025quantum, zhao2023maximum, yu2023robust, hangleiter2024robustly, rattacaso2023high, wang2015hamiltonian, li2020hamiltonian, granade2012robust, hincks2018hamiltonian, dutt2023active, ferrie2012adaptive, olsacher2025hamiltonian, bairey2020learning, chen2022learning, choi2022learning, kang2025enhanced,li2024heisenberg, mirani2024learning, zubida2021optimal, carrasco2021theoretical, rouze2024learning, anshu2020sample, hu2025demonstration, haah2022optimal, valenti2022scalable}.

The goal in HL research is to develop an efficient strategy to collect the maximal amount of information about the Hamiltonian using minimal resources, such as data or experimental runtime \cite{dutt2023active, dutkiewicz2023advantage}. Current strategies for Hamiltonian learning span a wide range of methods, including gradient-based optimization, as well as methods specifically tailored for robustness under noise to recover Hamiltonians from dynamical data \cite{wilde2022scalably, leng2025quantum, zhao2023maximum, yu2023robust, hangleiter2024robustly, rattacaso2023high}. Other works also reconstruct Hamiltonians by analysing specifically the many-body dynamics generated by interaction terms\cite{wang2015hamiltonian, bairey2019learning,li2020hamiltonian, gu2022practical} or by employing Bayesian optimization techniques \cite{granade2012robust, evans2019scalable, hincks2018hamiltonian, ferrie2012adaptive, olsacher2025hamiltonian}. Hybrid unitary-dissipative systems have also been investigated to capture the interplay between coherent and dissipative dynamics \cite{bairey2020learning, franca2022efficient, wiebe2014quantum, olsacher2025hamiltonian, chen2022learning}. Some frameworks have also integrated Neural Differential Equations to combine the expressiveness of neural networks with the underlying physics of Hamiltonians \cite{heightman2024solving, kang2025enhanced}. 

The highest achievable efficiency with respect to experimental resources is given by the so-called Heisenberg limit, while most strategies are typically bounded by the Standard Quantum Limit (SQL) \cite{Giovannetti_2006}. In this work we consider the interrogation time $t$ as the primary experimental resource. In scenarios where multiple interrogation times are required, for example when learning dynamical trajectories, we instead consider the total experiment time $T$, defined as the sum of the interrogation times applied to the same experimental configuration (same initial state and measurement basis), i.e.\ $T=\sum_i t_i$.Since our focus is on the scaling with respect to this time resource, we adopt a resource-accounting perspective that emphasizes the dependence on the interrogation time, while other resources are treated as linear since they are not the target of optimization in this work. Similar perspectives are used in other works with this focus, for example in \cite{dutkiewicz2023advantage}. The achievable estimation precision is bounded by the Cramér–Rao inequality $\Delta\theta \ge 1/\sqrt{\mathcal{I}_C}$ \cite{rao1992information}. In terms of the chosen resource, in our case the total experimental time $T_{\rm tot}$, the Heisenberg limit states that the estimation error $\Delta\theta$ of a Hamiltonian parameter $\theta$ can scale at best as $\Delta\theta=\mathcal{O}(T_{\rm tot}^{-1})$. In contrast, the Standard Quantum Limit predicts scaling $\Delta\theta=\mathcal{O}(T_{\rm tot}^{-1/2})$ \cite{nielsen2010quantum,Giovannetti_2006}.

Most existing methods that reach optimality with total experiment time rely in some degree on  dynamical control, either continuos or interleaved \cite{dutkiewicz2023advantage, hu2025ansatz, Huang2023Heisenberglimit, ni2024quantum, li2024heisenberg, dutt2023active}. Indeed Dutkiewicz  et. al \cite{dutkiewicz2023advantage} proved that achieving Heisenberg-limited learning (with respect to total experiment runtime) for quantum many-body Hamiltonians necessitates a continuous, high-fidelity control capability that scales with the desired precision \cite{dutkiewicz2023advantage}. However, implementing such continuous quantum many-body control fields remains a significant challenge \cite{luchnikov2024controlling}. Simulating and optimizing multi-qubit control pulses for interacting systems requires computational resources that grow exponentially with system size, and each fine-tuning iteration demands repeated, computational cost \cite{luchnikov2024controlling}. Moreover, the control sequences must be faster than the system's intrinsic coherence time, unless coherence times are correspondingly extended \cite{viola1998dynamical}. Another way to achieve optimality is through coherent, entangled measurements, which exploit quantum correlations to reduce the number of measurements required for state characterisation \cite{Giovanetti2004HeisenberglimitEntanglement}, bringing down the total experimental time. However, in practice coherent, entangled measurements remain challenging, since entangled states are fragile and readily destroyed by decoherence, demanding precise isolation, and also again large-scale control \cite{Giovanetti2004HeisenberglimitEntanglement,schlosshauer2004decoherence}.

In this work, we demonstrate how the naturally occurring quadratic short-time Fisher-information regime in quantum Hamiltonian learning can be made practically accessible while using only separable probe states, local measurements, and no dynamical control. Our approach relies on a formulation of the learning problem in which each experiment is initialized in a locally randomized product state, specifically \emph{spread states}, and measured with an ensemble of random Pauli-product bases. Operationally, this replaces repeated shots of identical experiments by ensembles of randomized probe-state and measurement configurations in order to improve the efficiency with respect to the time resource. Spread states activate the Hamiltonian eigenbasis and prevent state-induced suppression of spectral gaps, while randomized Pauli measurements reveal these spectral components in the measurement statistics. We show analytically that this combination activates the relevant spectral gaps in the measurement dynamics, allowing the Fisher information to exhibit quadratic short-time scaling over a finite and experimentally accessible temporal window. While quadratic Fisher scaling generically arises only in the infinitesimal limit $t \to 0$, spread-state preparation extends this regime to practical evolution times. Within this window the Fisher information grows as $\mathcal I(t)=\Omega(t^2)$, corresponding to Heisenberg-limit-consistent sensitivity with respect to the evolution time.\\
Remarkably, this regime can be made practically accessible using only local quantum operations applied at the state-preparation stage together with classical averaging over independently prepared probe configurations; once prepared, the system evolves under the Hamiltonian without dynamical intervention. In particular, recent no-go theorems show that for ETH Hamiltonians, achieving Heisenberg scaling in sequential estimation protocols that repeatedly probe a single experimental configuration requires continuous adaptive control~\cite{dutkiewicz2023advantage}. In that setting, control operations are applied during the evolution in order to actively modify the dynamics, progressively refining the effective Hamiltonian and rescaling the evolution time. By contrast, the protocol studied here does not modify the dynamics through control. Instead, it exploits the quadratic sensitivity that already arises naturally in short-time quantum evolution and extends the regime in which this scaling remains observable.
Secondly, we show that ensemble averaging over locally randomized product states also simplifies multiparameter estimation. As the ensemble size grows, the measurement sensitivities for each parameter become asymptotically independent (i.e., the Fisher information matrix becomes asymptotically diagonal), enabling all Hamiltonian parameters to be estimated in parallel from the same data set. We validate these results numerically for a range of disordered multi-qubit Hamiltonians in a black-box learning setting, confirming the learnability of full Hamiltonian matrices at once at asymptotically Heisenberg-like scaling under experimentally realistic constraints.\\
We further point out that once the learning rate for a single experiment grows super-linearly with the evolution time, the cumulative learning rate obtained from experiments performed at different evolution times does not necessarily exhibit the same scaling with the total experimental time. This is a purely mathematical effect arising from the properties of sums of power laws, which becomes relevant once the single-experiment scaling is non-linear. In particular, when experiments are performed at different evolution times, the cumulative learning rate depends on how the individual evolution times are distributed rather than only on their total duration. We analyse this effect and capture it through varied measurement-time scheduling.

The rest of this work is structured as follows: in Sec.~\ref{sec:theory} we briefly review the Fisher information and Cramér-Rao bound, before analysing the time evolution of the Fisher information for ensembles of spread states. This allows us to prove the asymptotic emergence of Heisenberg-limited scaling on such ensembles, and to prove that the Fisher information indeed diagonalises in the asymptotic limit. Next, Sec.~\ref{sec:numerical_results} describes the numerical experiments we conducted to practically demonstrate our theory away from the asymptotic limit; here we lay out the Hamiltonian models tested and show empirically that they can be learned at a rate better than the SQL. Additionally, we validate the recovery based results by looking at the Fisher information behaviour directly, allowing to exclude artefacts of the employed Hamiltonian recovery and reconstruction procedure.

\section{Theory}
\label{sec:theory}

\subsection{The Fisher information and the Cramér-Rao bound}
\label{sec:fisher_info}
Consider a quantum system governed by a Hamiltonian, whose spectral decomposition in the Pauli basis reads,
\begin{equation}
H(\theta) = \sum_j \theta_j P_j,\quad P_j \in \mathcal{P}_n = \bigl\{\sigma_{i_1}\otimes\sigma_{i_2}\otimes\cdots\otimes\sigma_{i_n}\;\big|\;\sigma_{i_k}\in\{I,X,Y,Z\}\bigr\},
\end{equation}
where \(\theta \in \mathbb{R}^d\) denotes a set of unknown parameters, and \(\mathcal{P}_n\)
is the \(n\)-qubit Pauli group. The primary objective of Hamiltonian learning is to estimate the coefficients \(\theta_j\) through experimental observations and, if necessary, to infer the Hamiltonian's operator structure. Learning scenarios are generally categorized based on whether prior information about the Hamiltonian's structure is available. When this is the case, only the coefficients in a given operator basis need to be tuned, and we arrive at the White-box scenario. On the other hand, without prior knowledge about the Hamiltonians structure, we have the Black-box scenario, where both the Hamiltonians coefficients and structure (often in the Pauli basis) must be uncovered \cite{heightman2024solving}. To understand the conditions in which HL parameters are optimally learnable, we can resort to the Fisher information \cite{petz2011introduction, watrous2018theory}.

The Fisher information tells us the amount of information about the parameter \(\theta\) that can be extracted from observation of a probabilistic observable, and is given by \cite{petz2011introduction},
\begin{equation}
\mathcal{I}(\theta) = \mathbb{E}\left[\left(\frac{\partial}{\partial \theta} \log p(x|\theta)\right)^2\right],
\end{equation}
where \(p(x|\theta)\) is the likelihood function \cite{Fisher_1925} \cite{rao1992information}. In the case of classical measurements, the general Fisher information becomes the classical Fisher information (CFI). The CFI associated with a probability distribution of measurement outcomes is given by,
\begin{equation}
\label{eq:classical_fisher}
F_C^{}(t) =  \sum_j \frac{1}{p_j^{}(t)} \left( \frac{\partial p_j^{}}{\partial \theta}{(t)} \right)^2,
\end{equation}
where \(p_j{}(t) = |\langle m_j | \psi_t^{(b)} \rangle|^2\) denotes the probability of obtaining a measurement outcome \(j\) at time \(t\) for a given initial state \(| \psi_t^{(b)} \rangle\) \cite{nielsen2010quantum,Giovannetti_2006}. The collected information can be used to estimate a governing Hamiltonian's true parameters up to a precision dependent on the amount of Fisher information. A fundamental limit on the precision of an estimator  \(\hat{\theta}\) is given by the Cramér-Rao (CR) bound,
\begin{equation}
\Delta \theta \geq \frac{1}{\sqrt{\mathcal{I}(\theta)}}
\end{equation}
where \(\Delta \theta\) is the estimation error and \(\mathcal{I}(\theta)\) is the Fisher information \cite{rao1992information}.

Hence, at the level of Fisher information, the SQL scaling of estimator in the parameter $\theta$ is linear, $F \propto t$, since this gives (by the CR bound) \(\Delta\theta = \mathcal{O}(t^{-1/2})\). In contrast, Heisenberg limited scaling at the level of Fisher information requires $F \propto t^2$, so that by the CR bound, we get \(\Delta\theta = \mathcal{O}(t^{-1})\). Let us now analyse the time evolution of the Fisher information under local randomness. In the following we will make use of the Bachmann-Landau notation to describe the asymptotic scaling behaviour of functions (see Appendix~\ref{app:landau_notation}).

\subsection{Time Evolution of Fisher Information under Local Randomness}
\label{sec:random_fisher_scaling}

 In this section we show how employing Haar-random local product
states and random Pauli product measurements extends the naturally
occurring quadratic scaling of the classical Fisher information
(Eq.~\ref{eq:classical_fisher}) at short times into a practically long
transient window.

For each measurement outcome \(j\), the probability admits a
short-time expansion
\begin{equation}
p_j(t,\theta)
=
p_j(0,\theta)
+
a_j(\theta)\,t
+
O(t^2),
\qquad t\to0 .
\end{equation}
Since the probe state and measurement do not depend on the Hamiltonian
parameters, the initial probabilities \(p_j(0,\theta)\) are independent
of \(\theta\), so that
\(
\partial_\theta p_j(0,\theta)=0
\).
Differentiating the expansion with respect to \(\theta\) therefore
yields
\begin{equation}
\partial_\theta p_j(t,\theta)
=
(\partial_\theta a_j(\theta))\,t
+
O(t^2),
\qquad t\to0 .
\end{equation}
If \(\partial_\theta a_j(\theta)\neq0\) for at least one outcome \(j\),
that is, if the chosen probe state and measurement activate a
non-vanishing parameter-sensitive contribution, then
\begin{equation}
|\partial_\theta p_j(t,\theta)| \ge c\,t = \Omega(t),
\qquad t\to0
\end{equation}
for some constant \(c>0\).\\
Using \(0<p_j(t,\theta)\le1\), the classical Fisher information
\eqref{eq:classical_fisher} satisfies
\begin{equation}
\mathcal I_C(t)
=
\sum_j
\frac{(\partial_\theta p_j(t,\theta))^2}{p_j(t,\theta)}
\ge
\frac{(\partial_\theta p_j(t,\theta))^2}{p_j(t,\theta)}
\ge
(\partial_\theta p_j(t,\theta))^2
=
\Omega(t^2),
\qquad t\to0 .
\end{equation}

To show how this quadratic scaling can persist over a longer and practical
temporal window, we examine the structure of the
measurement probabilities. Working in the Hamiltonian eigenbasis
\(H\ket{\lambda_k}=\lambda_k\ket{\lambda_k}\), the probability of
measurement outcome \(j\) can be written as
\begin{equation}
p_j(t)
=
\sum_{k,\ell}
c_{jk}^*c_{j\ell}\,
a_k a_\ell^*\,
e^{-i(\lambda_k-\lambda_\ell)t},
\label{eq:prob_sum_expansion}
\end{equation}
where \(a_k=\langle\lambda_k|\psi\rangle\) are the state amplitudes and \(c_{jk}=\langle\lambda_k|m_j\rangle\) denote the overlaps between the measurement basis and the Hamiltonian eigenbasis. Equation \eqref{eq:prob_sum_expansion} shows that the measurement dynamics is a superposition of oscillations with frequencies given by the spectral gaps \((\lambda_k-\lambda_\ell)\).  The oscillatory contributions in Eq.~\eqref{eq:prob_sum_expansion} are weighted by the coefficients \(c_{jk}^*c_{j\ell}a_k a_\ell^*\), so spectral components may be suppressed either by the state amplitudes \(a_k\) or by the measurement overlaps \(c_{jk}\). 
\begin{theorem}[Spectral window for quadratic Fisher scaling]
\label{thm:fisher_window}
If the coefficients \(c_{jk}^*c_{j\ell}a_k a_\ell^*\) activate all
relevant spectral gaps \((\lambda_k-\lambda_\ell)\), i.e.\
\(c_{jk}^*c_{j\ell}a_k a_\ell^* \neq 0\) for the corresponding pairs
\((k,\ell)\), then the classical Fisher information obeys
\[
\mathcal I_C(t)=\Omega(t^2)
\]
in the temporal window
\[
t \in
\left[
0,
\frac{\pi}{2\,\Delta\lambda_{\max}}
\right],
\]
where
\(
\Delta\lambda_{\max}=\max_{k\neq\ell}|\lambda_k-\lambda_\ell|
\)
is the largest nonzero spectral gaps.
\end{theorem}
\begin{proof}
    See 
    Appendix~\ref{app:spectral_fisher_window}.
\end{proof}

However, the spectral window only governs the observable dynamics if the corresponding spectral gaps are actually activated in the measurement signal. Structured initial states can therefore prevent parts of the Hamiltonian spectrum from contributing to the dynamics. To overcome this limitation we introduce \emph{spread states}. 

\begin{definition}[Spread State]
\label{dfn:state_spreading}
Let \(\ket{\psi_0}\) be any fixed \(n\)-qubit basis state of the Hamiltonian of interest. For each qubit \(j\), draw Euler angles
\((\xi_j,\chi_j,\phi_j)\) i.i.d.\ from the \(SU(2)\) Haar measure and define
\begin{equation}
U_{\rm spread}
=
\bigotimes_{i=1}^n
R_z(\xi_j)R_y(\chi_j)R_z(\phi_j),
\qquad
\ket{\psi_{\rm spread}} = U_{\rm spread}\ket{\psi_0}.
\end{equation}
\end{definition}

These states are similar to classical shadows used in quantum state tomography, however notice that we are performing one-local Haar random rotations \textit{before} the Hamiltonian action. As a consequence spread states remove the state-induced suppression of oscillatory terms in the measurement probabilities. By Lemma~\ref{lem:generic_spectral_activation}, the expansion of a spread state in the Hamiltonian eigenbasis generically contains nonzero amplitudes for all eigenstates,
\begin{equation}
\ket{\psi_{\rm spread}}
=
\sum_k a_k\ket{\lambda_k},
\qquad
\mathbb{P}_{\{U_r\}}\!\left[a_k\neq0\ \forall k\right]=1 .
\label{eq:spread_state_expansion}
\end{equation}
Because the spread-state construction applies independent Haar-random rotations to each qubit, the resulting randomization acts locally and therefore does not deteriorate with increasing system size.

The remaining filtering of spectral components is therefore determined by the measurement overlaps \(c_{jk}\). For a fixed measurement basis, some overlaps may still vanish, suppressing particular terms \(c_{jk}^*c_{j\ell}a_k a_\ell^*\) in Eq.~\eqref{eq:prob_sum_expansion}. To mitigate this effect we sample an ensemble of random local Pauli product measurement bases. Different measurement realizations generically produce different overlap structures \(c_{jk}\), so that across the ensemble a larger set of spectral frequencies \((\lambda_k-\lambda_\ell)\) becomes visible in the measurement signal.

In principle, a single spread state already activates the Hamiltonian eigenbasis with probability one, so that the quadratic short-time scaling of the Fisher information can be accessed using an ensemble of measurement bases alone. Sampling multiple spread states is therefore not required for the existence of the quadratic regime, but improves the typical spectral weights and hence the visibility of the effect.

Thus we have shown that spread states (Theorem~1) extend the naturally occurring quadratic scaling of the Fisher information beyond the infinitesimal-time regime into a finite transient window. This occurs because spread states effectively activate the relevant spectral components of the Hamiltonian. Consequently, within this transient window one can achieve Heisenberg-limited learning, since quadratic scaling of the Fisher information corresponds directly to Heisenberg-limit precision via the Cramér–Rao bound.

\subsection{Fisher Information Diagonalization}
\label{sec:fisher_isotropy}
When learning a Hamiltonian with multiple components \( \theta = (\theta_1, \dots, \theta_d) \), different components can simultaneously influence the measurement outcome in correlated ways. In such cases, adjusting or estimating one parameter effectively depends on the values of other parameter. Hence, many HL strategies must be carefully tailored to isolate or selectively probe each parameter one by one, which usually requires some kind of structural prior such as locality sparsity, commutativity or optimal control \cite{Huang2023Heisenberglimit, yu2023robust, dutkiewicz2023advantage}. Ideally, one would like a learning strategy where all parameters can be estimated independently and simultaneously. This is possible due to the randomness of state spreading effectively diagonalising the Fisher information, as detailed in this section.

Another key consequences of using spread states is that they produce measurement statistics that are sensitive to all components of the Hamiltonian at leading order in time, regardless of structure. This follows from the fact that independently sampled local Haar unitaries almost surely generate non-vanishing overlaps with all Pauli terms. Formally, for a product state
\(
\ket{\psi_0} = \bigotimes_{j=1}^n U_j \ket{\phi_j},
\)
where each \(U_j \in \mathrm{SU}(2)\) is independently drawn from the Haar measure, the short-time measurement probability  for a pauli-product basis outcome, i.e., a bitstring $b \in \{0,1\}^n$,
\begin{equation}
p(b, t) \approx |\langle b | \psi_0 \rangle|^2 - 2t\, \mathrm{Im} \!\left[ \sum_a \theta_a \langle b | \psi_0 \rangle^* \langle b | P_a | \psi_0 \rangle \right]
\end{equation}
where, with probability approaching unity,
\(\langle b | \psi_0 \rangle \neq 0\) and \(\langle b | P_a | \psi_0 \rangle \neq 0\) for all \(a\), ensuring that no term $P_a$, in the Hamiltonian \( H = \sum_a \theta_a P_a \), is missed at first order in time. Hence \(p(b, t)\) is generically sensitive to every Hamiltonian parameter \(\theta_a\) at first order in time, as formally derived in Lemma~\ref{lemma:generic_sensitivity}.

Nevertheless, each individual spread initial state can still produce parameter interdependence in the multi-parameter Fisher information (also given as the Fisher information matrix \cite{liu2020quantum}: 
\(\mathcal{I}_{jk} = \sum_i p_i^{-1} (\partial_{\theta_j} p_i)(\partial_{\theta_k} p_i)\)), due to cross terms of the form
\begin{equation}
    \frac{1}{p_i} \frac{\partial p_i}{\partial \theta_j} \frac{\partial p_i}{\partial \theta_k} \neq 0 \quad \text{for } j \neq k,
\end{equation}
as discussed in the general theory of quantum multiparameter estimation~\cite{liu2020quantum}.\\

\noindent
But, these correlations vanish when averaging over an ensemble of spread states. This follows from the isotropy of the local Haar measure, which ensures that expectation values of distinct Pauli strings are statistically independent with zero mean but finite variance. As a result, the cross-correlations between different parameter derivatives vanish in expectation,  while the finite variances of individual terms preserve the diagonal entries.  Consequently, the ensemble-averaged Fisher information matrix becomes diagonal,  with uncorrelated yet non-vanishing sensitivities in the limit of many random instances. This is formalized in the following theorem:

\begin{theorem}[Fisher Information Diagonalization]
\label{thm:fisher_diagonalization}
Let \( H(\theta) \) be a  $k$-local Hamiltonian parametrized by $d$ real parameters \( \theta = (\theta_1, \dots, \theta_d) \), so that \( H(\boldsymbol{\theta}) = \sum_j \theta_j P_j \), and let each experiment prepare an initial state via local Haar-random unitaries applied to a fixed product state, with measurements in a fixed Pauli product basis. Let \( \mathcal{I}_r(\theta) \) be the Fisher information matrix resulting from the \( r \)-th experiment, for \( r = 1, \dots, R \).  At leading order in time, the ensemble-averaged Fisher information matrix converges to a diagonal form,
\begin{equation}
\lim_{R \to \infty} \left( \frac{1}{R} \sum_{r=1}^R \mathcal{I}_r(\theta) \right) = \mathrm{diag}(c_1, \dots, c_d) \quad \text{with } c_j > 0.
\end{equation}
\end{theorem}

\begin{proof}
    See 
    Appendix~\ref{lemma:isotropy_spreading}.
\end{proof}

The Fisher-information diagonalization described above does not rely on the Hilbert-space dimension being small. 
For $k$-local Hamiltonians with bounded connectivity, the commutator structure entering the Fisher information involves only $O(1)$ local operator contributions. Consequently, the ensemble-averaged Fisher matrix and the diagonalization mechanism remain independent of the Hilbert-space dimension. The locality-based argument establishing this property is provided in Appendix~\ref{lemma:isotropy_spreading}.

Hence, using ensembles of spread states removes the need to design separate experiments for each Hamiltonian term, and therefore allows all Hamiltonian components to be learned simultaneously from the same dataset, without needing to isolate or separate them one at a time. This effect enables to learn all parameters together from one dataset within the Heisenberg-limited scaling window, also enabled by spread state ensembles as shown in the previous section.

\subsection{Degradation of non-Linear Learning Rates in Multiple-Timestamp Scenarios}
\label{sec:quadratic_scheduling}

In Hamiltonian learning based on trajectories, data can be collected at multiple time points\cite{dutkiewicz2023advantage, wilde2022scalably, heightman2024solving}. The total duration of the experiment is then given by 
\begin{equation}
    T_{\mathrm{tot}} = \sum_k t_k,
\end{equation}
and we are interested in how the total Fisher information grows with this total time.

In the Sec.~\ref{sec:random_fisher_scaling}, we showed that a Heisenberg limited regime can emerge with respect to the particular evolution time $t \ll 1$.
However, if the Fisher information for a single evolution time $t$ scales non-linearly, that is,
\begin{equation}
    \mathcal{I}(t) \propto t^{\nu} \quad \text{with } \nu \neq 1,
\end{equation}
then the cumulative Fisher information obtained from multiple measurement times does not necessarily scale with the total experimental duration $T_{\mathrm{tot}}$ in the same way. 
In principle, one might expect the total Fisher information to scale with the total experiment time as
\begin{equation}
    \mathcal{I}_{\mathrm{tot}} \propto T_{\mathrm{tot}}^{\nu} = \bigl(\sum_k t_k\bigr)^{\nu},
\end{equation}
but in practice the Fisher information accumulates additively across measurement times as
\begin{equation}
    \mathcal{I}_{\mathrm{tot}} = \sum_k \mathcal{I}(t_k) \propto \sum_k t_k^{\nu}.
\end{equation}
For $\nu \neq 1$, these two quantities satisfy the inequality
\begin{equation}
    \sum_k t_k^{\nu} \le \left(\sum_k t_k\right)^{\nu},
\end{equation}
with equality only when all but one $t_k$ vanish. 
Hence, the total Fisher information generally scales more slowly with $T_{\mathrm{tot}}$ than one might expect from the single-time behaviour $\mathcal{I}(t) \propto t^{\nu}$.
Nevertheless, by choosing the measurement times appropriately, this discrepancy can be reduced. 
Hence, we introduce a simple family of non-uniform measurement schedules, 
where measurements are performed at times
\begin{equation}
    t_k = \Delta t\, k^{\alpha},
\end{equation}
for some exponent $\alpha > -1$ to be optimised. 
We generally write the Fisher information scaling with time as 
\(\mathcal{F}(t) = \Theta(t^{\gamma_0})\), 
so that \(\gamma_0 = 2\) corresponds to Heisenberg-limited scaling. 
For our non-uniform sampling schedule, we take
\begin{equation}
t_k = \Delta t\,k^\alpha,
\qquad k = 1,\dots,m_t,
\end{equation}
with \(\Delta t>0\) and \(\alpha>-1\). 
This leads us to the following result.
\begin{proposition}[Cumulative Fisher-Information Scaling]
\label{prop:cumulative_scaling}
If \(t_k = \Delta t\,k^\alpha\) for \(k = 1,\dots,m_t\), and each \(t_k\) lies within the regime where the Fisher information obeys 
\(\mathcal{F}(t_k) = \Theta(t_k^{\gamma_0})\) 
(with \(\gamma_0=2\) under Heisenberg-limited scaling), 
then the total Fisher information
\begin{equation}
I_{\rm tot} = \sum_{k=1}^{m_t} \mathcal{F}(t_k),
\end{equation}
scales with the total experimental time 
\(T_{\rm tot} = \sum_{k=1}^{m_t} t_k\) as
\begin{equation}
I_{\rm tot} = \Theta\!\bigl(T_{\rm tot}^{\,p}\bigr),
\qquad
p = \frac{\alpha \gamma_0 + 1}{\alpha + 1} + O(m_t^{-1}),
\end{equation}
where \(p\) is the effective scaling exponent including finite-$m_t$ corrections.
\end{proposition}
\noindent\textit{Proof.} See Appendix~\ref{app:cumulative_fisher_information_scaling}.
With the proposed scheduling, it is possible to tune $\alpha$ (while keeping $\gamma_0 = 2$) so that one can continuously move from uniform sampling with $\alpha = 0$, giving 
\(
I_{\rm tot} = \Theta(T_{\rm tot}),
\)
corresponding to the standard quantum limit, to strongly non-uniform sampling with $\alpha \to \infty$, yielding
\(
I_{\rm tot} = \Theta(T_{\rm tot}^p), \qquad p \to 2.
\)
Thus, we asymptotically recover Heisenberg-limited scaling with respect to the total experimental time as $\alpha \to \infty$.

\section{Numerical Demonstrations}
\label{sec:numerical_results}
Having established the theoretical results in Sec.~\ref{sec:theory}, we now provide practical demonstrations of the proposed learning protocol using simulated experiments. The theoretical analysis shows that randomized spread states together with randomized Pauli measurements activate the Hamiltonian spectrum and allow the quadratic short-time Fisher-information regime to be accessed in practice. In this section we illustrate how these mechanisms translate into concrete Hamiltonian-learning performance in small-system settings. In particular, we show empirically that the protocol can surpass the SQL and enable simultaneous estimation of all Hamiltonian parameters using only local product states and local measurements.

Our numerical demonstration is given by simulated experiments performed on 5 qubits. The mechanisms analysed in Sec.~\ref{sec:theory} persist independently of the Hilbert-space dimension. In particular, spectral activation arises from locally Haar-random spread states acting independently on each qubit, while the Fisher-information structure underlying the protocol depends only on the locality of the Hamiltonian terms. As a result, the behaviour illustrated in the simulations reflects the protocol mechanism rather than a special property of small system sizes. We detail the Hamiltonian models used in Sec.~\ref{sec:model}, followed by formulating the recovery error of interest and relating it to the total experiment time in Sec.~\ref{sec:recovery_error_and_total}. Sec.~\ref{res:accessing_quadratic_regime} then presents the first experiment, where we fix a time-stamp schedule and vary the spread-state ensemble size. Here we analyse how ensemble averaging improves the simultaneous learning of all Hamiltonian parameters. In the second experiment, presented in Sec.~\ref{sec:numerical_emergent_coherence}, we fix a spread-state ensemble that already yields sub-SQL learning with evolution time $t$ and vary the time-stamp scheduling to study how the learning behaviour depends on the total experiment time $T_{\rm tot} = \sum_k t_k$.

 In addition to the recovery experiments described above, we perform Fisher-information diagnostics in Secs.~\ref{sec:fisher_diagnostics_theory} and \ref{sec:fisher_diagnostics_experiments} to analyse the intrinsic information structure of the measurement protocol independently of the Hamiltonian reconstruction procedure. Because reconstruction performance can be influenced by numerical optimization effects such as convergence behaviour or conditioning of the estimation problem, analysing the Fisher information directly from the simulated measurement probabilities allows us to verify that the quadratic short-time sensitivity regime predicted in Sec.~\ref{sec:random_fisher_scaling} already emerges in the measurement statistics and follows the expected cumulative scaling with the total experiment time. In addition, we analyse the structure of the Fisher information matrix to validate the diagonalization mechanism predicted by Theorem~\ref{thm:fisher_diagonalization}. The theorem establishes that averaging over spread states drives the Fisher matrix towards a diagonal form in the large-ensemble limit. The fisher diagnostics quantifies how this decorrelation emerges for the finite spread-state ensembles used in the experiments by directly at the ratio between diagonal and non-diagonal fisher matrix entries. We accompany these investigations with system size scaling, to also validate how these mechanisms behave as the system size increases, confirming that both the Fisher scaling and the spread-state-induced diagonalization predicted by the theory remain effective independent Hilbert-space dimensions.

\subsection{Hamiltonian Models}
\label{sec:model}

We choose the following three anisotropic, disordered Heisenberg model Hamiltonians, with local transverse fields as ground truth Hamiltonians to be recovered: 
\begin{subequations}
\label{eq:xyz_hamiltonians}
\begin{align}
H_{\mathrm{XYZ}} \label{eq:first_test_ham}
  &= \sum_{i=1}^{N-1} \left(J^x_i X_i X_{i+1} + J^y_i Y_i Y_{i+1} + J^z_i Z_i Z_{i+1}\right)
   + \sum_{i=1}^N h_i\,X_i \\[1.5ex]
H_{\mathrm{XYZ2}} \label{eq:second_test_ham}
  &= \sum_{i=1}^{N-1} \left(J_i^x X_i X_{i+1} + J_i^y Y_i Y_{i+1} + J_i^z Z_i Z_{i+1}\right)
   + \sum_{i=1}^{N} \left(h_i^x X_i + h_i^y Y_i + h_i^z Z_i\right) \nonumber \\
  &\quad + \sum_{i=1}^{N-2} \left(K_i^x X_i X_{i+2} + K_i^y Y_i Y_{i+2} + K_i^z Z_i Z_{i+2}\right) \\[1.5ex]
H_{\mathrm{XYZ3}} \label{eq:third_test_ham}
  &= \sum_{i=1}^{N-1} \left(J_i^x X_i X_{i+1} + J_i^y Y_i Y_{i+1} + J_i^z Z_i Z_{i+1}\right)
   + \sum_{i=1}^{N} \left(h_i^x X_i + h_i^y Y_i + h_i^z Z_i\right) \nonumber \\
  &\quad + \sum_{i=1}^{N-2} \left(K_i^x X_i X_{i+1} X_{i+2} + K_i^y Y_i Y_{i+1} Y_{i+2} + K_i^z Z_i Z_{i+1} Z_{i+2}\right)
\end{align}
\end{subequations}
where each coupling axis \((J^x_{ij}, J^y_{ij}, J^z_{ij})\) is independently disordered, and each qubit experiences a distinct, anisotropic local field. They gradually increase in complexity due to the presence of next-nearest-neighbour interactions in Eq.~(\ref{eq:second_test_ham}), and the third-order coupling in Eq~(\ref{eq:third_test_ham}), given as \((K^x_{ij}, K^y_{ij}, K^z_{ij})\) in both cases.

The coefficients in front of each interaction define the ground-truth parameters to be recovered. 
For each interaction (nearest-neighbor, next-nearest, or three-body), we draw anisotropic couplings as 
\(J^\nu_\mu \sim U(-1,1)\),  where \(U(a,b)\) denotes the uniform distribution on the interval \([a,b]\), with \(\nu \in \{x, y, z\}\) labelling the axis, and \(\mu\) the interacting sites. Gaussian disorder is then added as 
\(J^\nu_\mu \leftarrow J^\nu_\mu + \delta J^\nu_\mu\) with 
\(\delta J^\nu_\mu \sim \mathcal{N}(0,\sigma^2)\) and \(\sigma = 0.1\). 
The same sampling and disorder procedure is applied to the higher-order coupling coefficients \(K^\nu_\mu\). 
For each site \(i\), we additionally draw independent local field components 
\(h_i^\nu \sim U(-1,1)\) for \(\nu \in \{x, y, z\}\), to model fully anisotropic site-specific fields.


In addition to the above, we include a simplified XXZ-type model that serves as a special case to test recovery in the gapless regime. This model contains only nearest-neighbor interactions and no local fields, and its Hamiltonian reads,
\begin{equation}
\label{eq:xxygl_ham}
H_{\text{XXZ}} = \sum_{i=1}^{N-1} \left(X_i X_{i+1} + Y_i Y_{i+1} + \Delta\, Z_i Z_{i+1}\right),
\end{equation}
where \(\Delta \sim U(-0.5, 0.5)\) is the anisotropy parameter, sampled independently for each instance. This model lacks a spectral gap in the range $|\Delta| \leq 1$ \cite{takahashi1999thermodynamics}.

\subsection{Hamiltonian Estimator}
\label{sec:estimator}
To translate measurement data into Hamiltonian estimates, we employ a maximum-likelihood estimator constructed from the measurement statistics generated in the simulated experiments. A detailed description of the dataset generation and implementation is provided in Appendix~\ref{sec:recovery}; here we summarize the estimator used in the numerical demonstrations.

For each ensemble configuration, defined by the spread-state preparation, time-evolution schedule, and random Pauli-product measurements, we collect bit-string outcomes forming a dataset \(D\) of size \(|D| = R J m_t S\). The entries of \(D\) are indexed by \((r,j,k,s)\), corresponding to the spread state, measurement basis, time stamps, and repetition shot number, with measurement outcomes \(b_{rjks}\).

We parametrize a candidate Hamiltonian \(\hat H(\theta)\) by its independent matrix elements.
For a given parameter vector \(\theta\), the measurement probabilities predicted by the model are
\begin{equation}
P\bigl(b_{rjks}\mid \psi^{(r)}_0, p_j, t_k,\theta\bigr)
=
\bigl|\langle b_{rjks}\mid e^{-i\hat H(\theta)t_k}\mid\psi^{(r)}_0\rangle\bigr|^2 .
\end{equation}

The Hamiltonian estimate is obtained by minimizing the negative log-likelihood of the observed measurement outcomes,
\begin{equation}
\mathcal{L}_D(\theta)
=
-\frac{1}{R J m_t S}
\sum_{r,j,k,s}
\log P\bigl(b_{rjks}\mid \psi^{(r)}_0, p_j, t_k,\theta\bigr),
\end{equation}
which yields the Hermitian matrix \(\hat H(\theta)\) that best reproduces the measurement statistics. The resulting Hamiltonian estimate is then used in the recovery experiments to compute the reconstruction error defined in the following Sec.~\ref{sec:recovery_error_and_total} in Eq.~\eqref{eq:evaluation_main}. The used black-box Hamiltonian recovery method is given in Appendix~\ref{sec:recovery}. Here we employed a black-box Hamiltonian learning method to ensure that any biases arising from a priori knowledge do not influence the observed learning rate. The code for our method is available online and open-source \cite{githubrepo}. 

We note that the details of the method are not essential and that a variety of methods based on trajectories and maximum likelihood are applicable here and can get the benefits of spread states and time-scheduling \cite{wilde2022scalably, hangleiter2024robustly,zhang2017quantum, zhao2023maximum, zhao2024learning, kang2025enhanced, wang2015hamiltonian, li2020hamiltonian, pastori2022characterization, olsacher2025hamiltonian, han2021tomography, valenti2022scalable, yu2023robust, liu2017quantum, heightman2024solving}.

We emphasize that the theoretical results derived in Sec.~\ref{sec:theory} concern the Fishe information structure of the measurement protocol and therefore do not depend on the particular estimator used. In the numerical demonstrations we first employ a Hamiltonian recovery method in order to provide a practical end-to-end illustration of how the proposed protocol can be used to infer Hamiltonian matrices from measurement data.

To validate that the observed learning behaviour arises from the intrinsic information content of the measurement statistics rather than from artefacts of the reconstruction procedure, we later complement these recovery experiments with a direct Fisher-information diagnostic (Sec.~\ref{sec:fisher_diagnostics_theory}). In this diagnostic we intercept the analysis at the level of the measurement probabilities and evaluate the Fisher information directly from their parameter dependence, before any Hamiltonian inference is performed. This allows us to verify that the observed scaling behaviour originates from the build-up of Fisher information in the measurement statistics themselves, rather than from properties of the recovery algorithm.

\subsection{Recovery Error and the Total Experiment Time}
\label{sec:recovery_error_and_total}
For each ensemble configuration, and its measurement outcomes, we recover a Hamiltonian Matrix and compute the recovery error as the mean absolute element-wise difference,
\begin{equation}
\label{eq:evaluation_main}
\varepsilon = \frac{1}{{d}^2}\sum_{i,j} \bigl| H^{\mathrm{true}}_{ij} - \hat{H}(\theta)_{ij}\bigr|,
\end{equation}
where \(H^{\mathrm{true}}\) and \(\hat H(\theta)\) denote the true and recovered Hamiltonian matrices, respectively,  and \(d = 2^n\) is the Hilbert-space dimension, so that \(d^2\) is the total number of Hamiltonian matrix elements.

The analytical connection between the experimentally observable scaling in accuracy  of the Hamiltonian recovery (given by Eq.~\ref{eq:evaluation_main}) and our theoretical description is derived as follows.
Following Proposition~\ref{prop:cumulative_scaling} and the Cramér–Rao bound, the recovery error $\varepsilon$  obeys the following behaviour (assuming Heisenberg limit with $t$: $\gamma_0 \approx 2$):
\begin{equation} \label{eq:err_1}
    \varepsilon \sim I_c(\theta)^{-1/2},
    \quad
    I_c(\theta) = \Theta\!\bigl(T_{\rm tot}^p\bigr),
    \quad
    p = \,\frac{\alpha\,\gamma_0 + 1}{\alpha + 1} + O(m_t^{-1})\quad \gamma_0 = 2,
\end{equation}
which implies
\begin{equation} \label{eq:err_2}
    \varepsilon \propto T_{\rm tot}^{-\beta},
    \qquad
    \beta(T_{\rm tot}) = \frac{p}{2},
    \qquad
    T_{\rm tot} = \sum_{k=1}^{m_t} t_k ,
\end{equation}
where $T_{\rm tot}$ is the total experiment time, given multiple time stamps. This allows us to characterise the recovery error in terms of the total experiment time. We derived $\beta$ as the experimentally accessible error-scaling exponent, and $\alpha$ is the scheduling parameter controlling the distribution of measurement times (introduced in section \ref{sec:quadratic_scheduling}). 
Together, Eqs.~(\ref{eq:err_1})–(\ref{eq:err_2}) provide the theoretical benchmark to interpret the numerical results.

\subsubsection{Illustration of Beyond-SQL Hamiltonian Learning}
\label{sec:representative_examples}

To illustrate the behaviour of the learning protocol, we first present representative examples for each of the considered Hamiltonian models in Sec.~\ref{sec:model}, in Fig.~\ref{fig:break_SQL_families}. In these examples we fix a moderately large ensemble of $R=32$ spread states, measured in $25$ random Pauli-product bases at $m_t=8$ evolution times $t_k = \Delta t\,k^\alpha$ with $\Delta t=0.01$ and $\alpha=1.0$.

For all Hamiltonian families considered in Sec.~\ref{sec:model}, the reconstruction error of the full Hamiltonian matrix decreases with the total experiment time $T_{\rm tot} = \sum_k t_k$ at a rate surpassing the standard quantum limit. The observed behaviour is consistent with the theoretically predicted scaling of the recovery error $\varepsilon \propto T_{\rm tot}^{-\beta}$ derived in Sec.~\ref{sec:recovery_error_and_total}, assuming quadratic single-parameter Fisher scaling in the short-time regime.

These examples demonstrate that ensemble-averaged randomized measurements enable beyond-SQL Hamiltonian learning across a variety of many-body Hamiltonian families.

\begin{figure}[ht!]
  \centering
  \begin{subfigure}[t]{0.45\textwidth}
    \includegraphics[width=\textwidth]{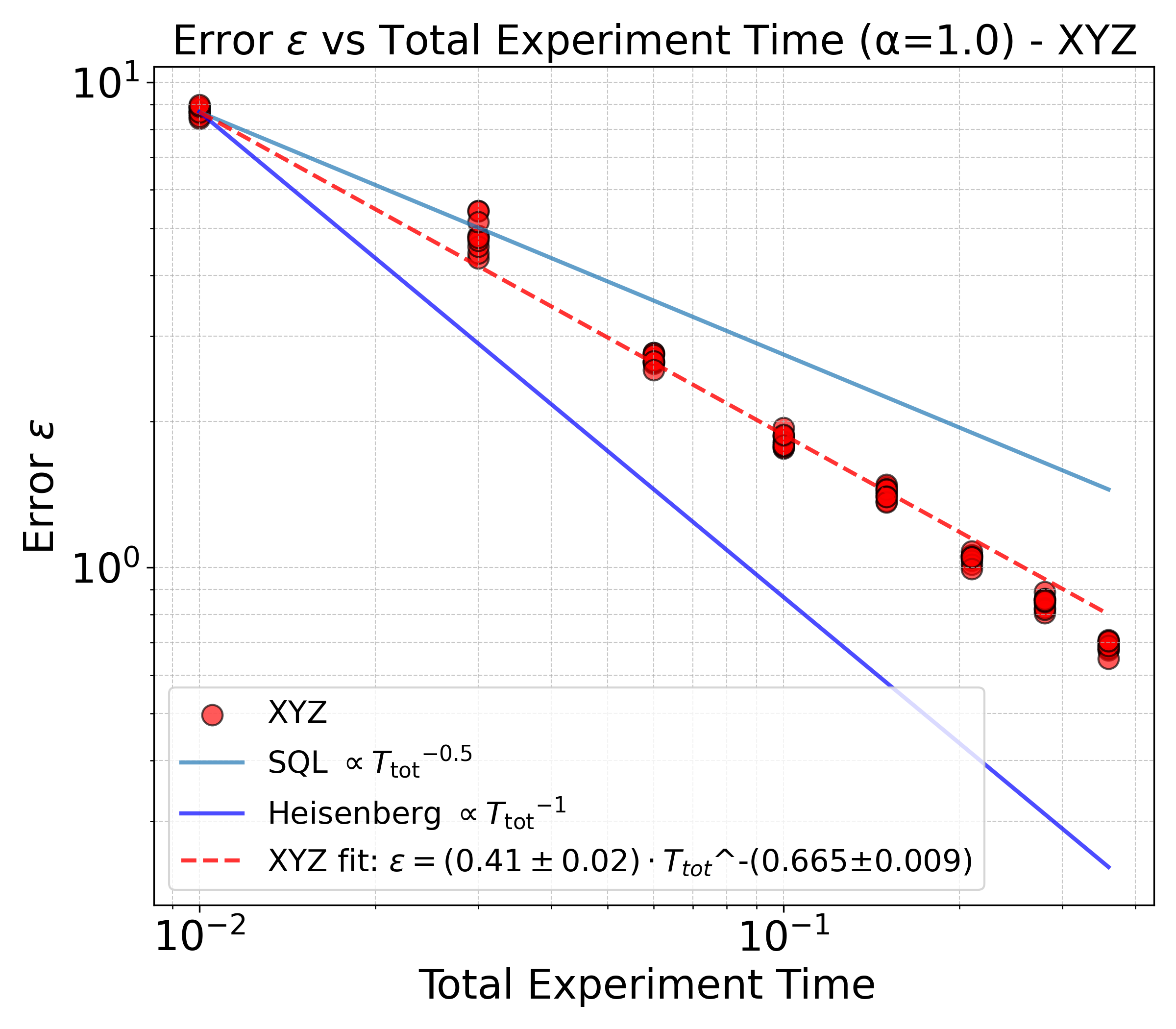}
    \caption{XYZ Hamiltonian}
    \label{fig:family1_alpha1}
  \end{subfigure}
  \hfill
  \begin{subfigure}[t]{0.45\textwidth}
    \includegraphics[width=\textwidth]{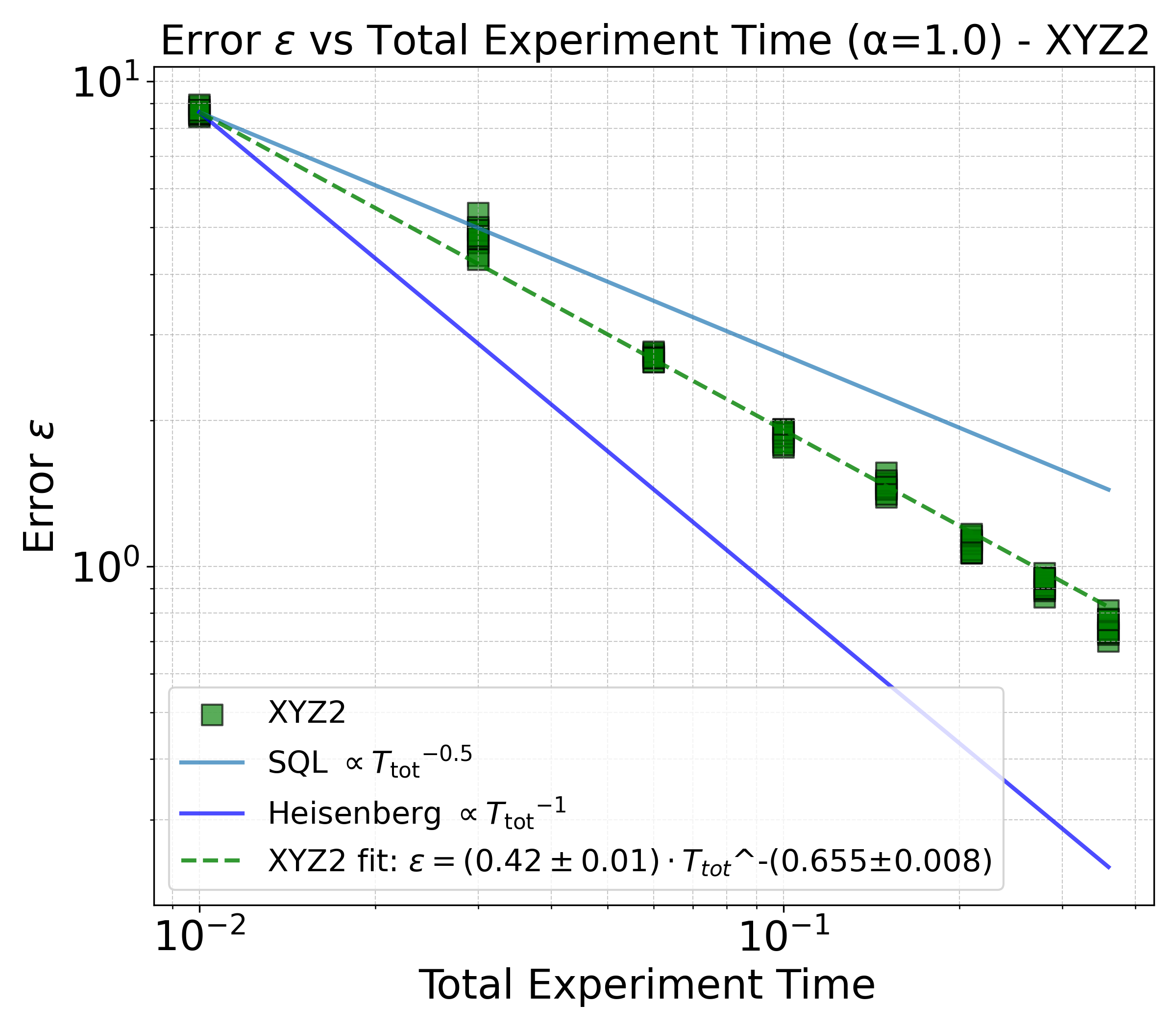}
    \caption{XYZ2 Hamiltonian}
    \label{fig:family2_alpha1}
  \end{subfigure}

  \vspace{0.5em}
  \begin{subfigure}[t]{0.45\textwidth}
    \includegraphics[width=\textwidth]{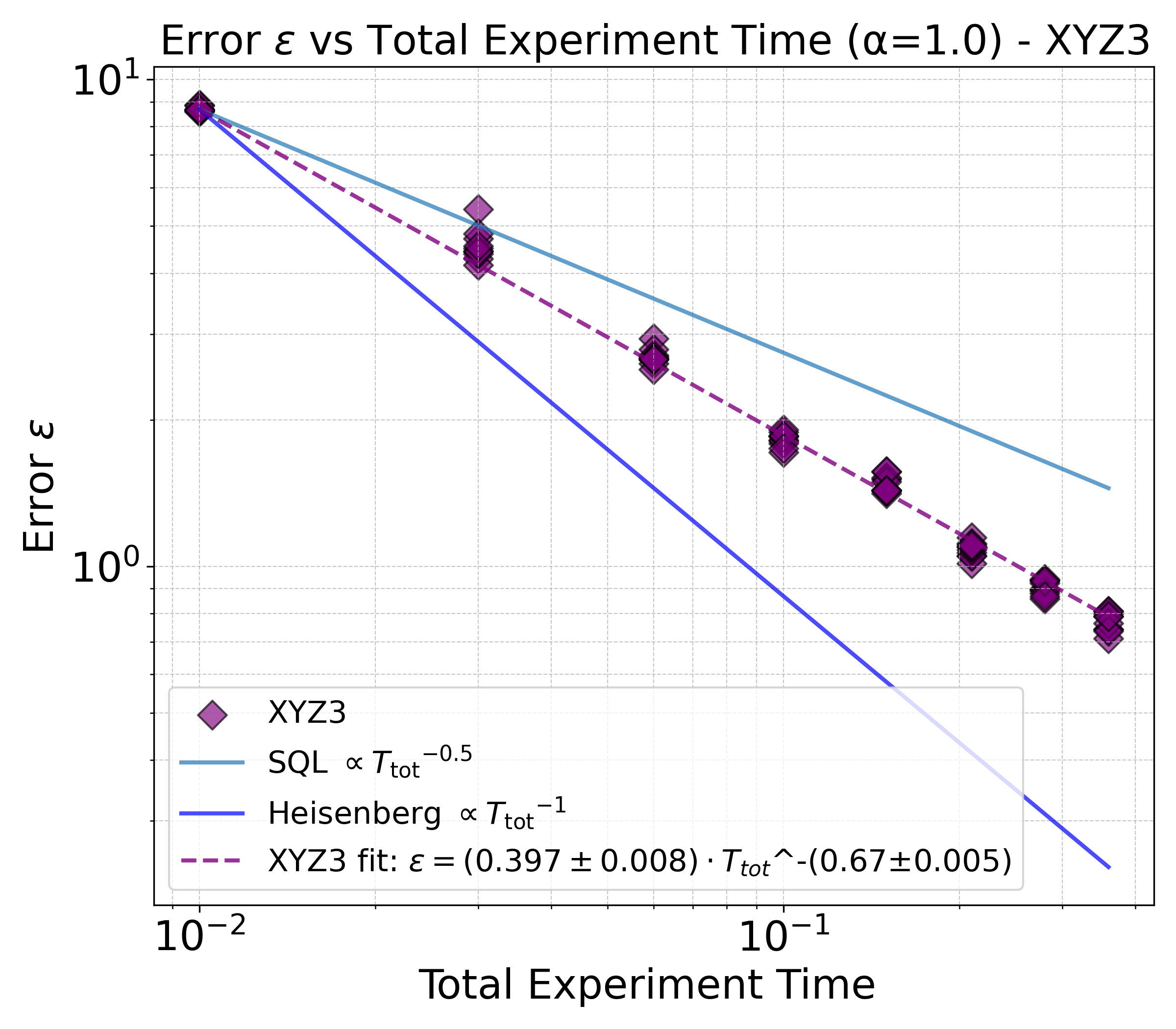}
    \caption{XYZ3 Hamiltonian}
    \label{fig:family3_alpha1}
  \end{subfigure}
  \hfill
  \begin{subfigure}[t]{0.45\textwidth}
    \includegraphics[width=\textwidth]{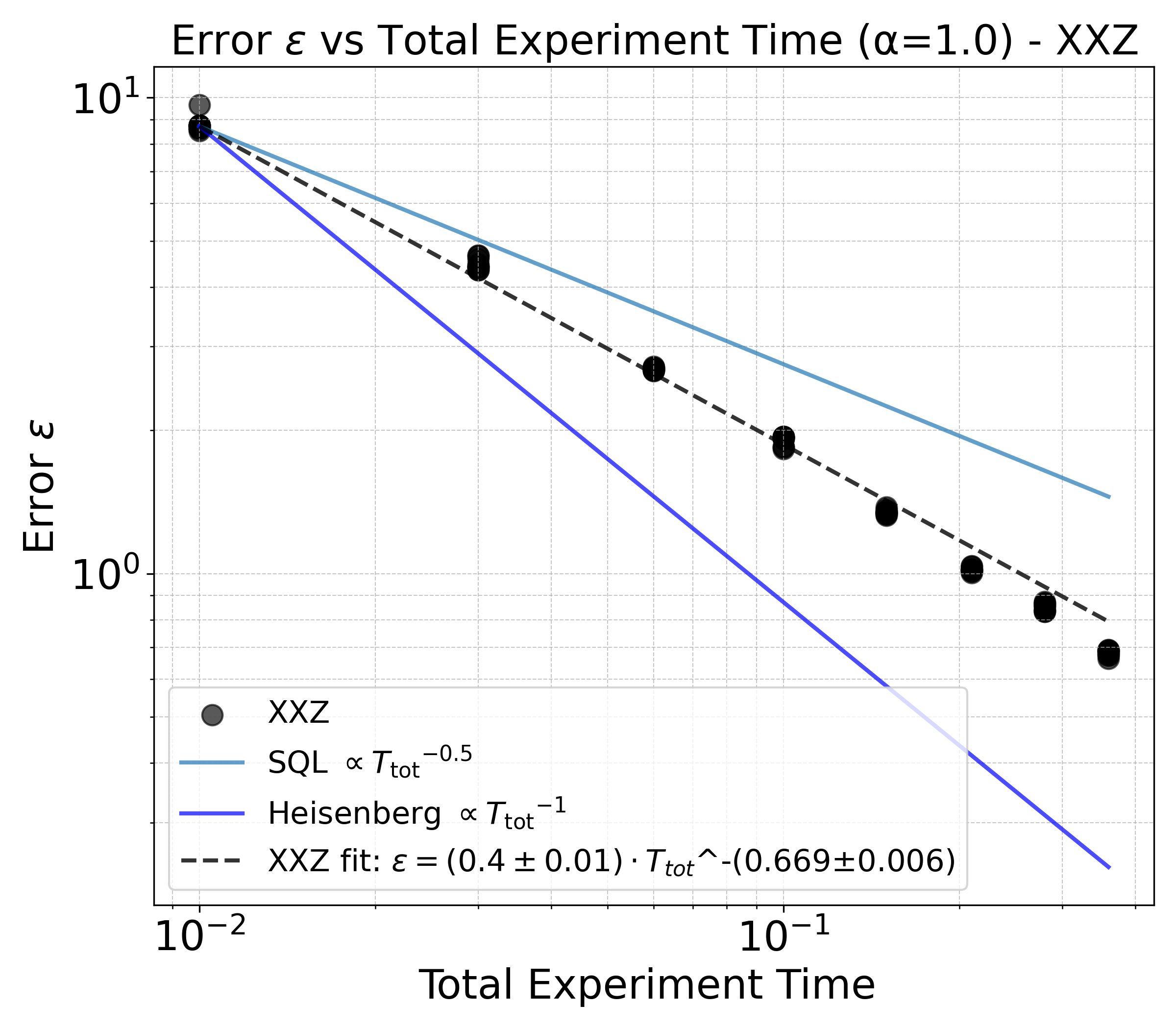}
    \caption{XXZ Hamiltonian}
    \label{fig:family4_alpha1}
  \end{subfigure}
  \caption{
Reconstruction error $\varepsilon$ versus total experiment time $T_{\rm tot}$ for four representative Hamiltonian families at $\alpha = 1.0$. 
Each data point is given with ten random Hamiltonian realizations. 
Across all Hamiltonians, the error decays as $\varepsilon \propto T_{\rm tot}^{-\beta}$ with $\beta \approx 0.66$, surpassing the standard quantum limit and consistent with Heisenberg-limited single-probe scaling.}
  \label{fig:break_SQL_families}
\end{figure}

\subsection{Hamiltonian Learning with Ensembles of Spread Initial States}
\label{res:accessing_quadratic_regime}

In this section, we demonstrate how increasing the number of spread states improves the stability of Hamiltonian learning and enables the simultaneous recovery of all Hamiltonian parameters from the same dataset, even away from the asymptotic limit. In particular, averaging over spread states leads to an effective diagonalization of the Fisher information matrix, as predicted by Theorem~\ref{thm:fisher_diagonalization} in Sec.~\ref{sec:fisher_isotropy}. This suppresses cross-correlations between parameters so that the quadratic single-parameter sensitivity regime becomes simultaneously accessible for all Hamiltonian parameters.

To examine this effect, we generated ensembles of $R \in \{2,4,8,16,32,64\}$ spread states according to Def.~\ref{dfn:state_spreading}, measured each at $m_t=15$ evolution times $t_k = \Delta t\,k^\alpha$ with $\Delta t=0.01$ and fixed $\alpha=1.0$, and performed one-shot measurements in 25 random Pauli-product bases.  The Hamiltonian parameters were recovered via maximum-likelihood estimation for the Hamiltonian Matrix (see Sec.~\ref{sec:recovery}). For each ensemble size $R$, we fitted the reconstruction error $\varepsilon$ against total experiment time $T_{\rm tot} = \sum_k t_k$ to extract the scaling exponent $\beta$, and then studied its dependence on $R$.

Figure~\ref{fig:scaling_state_spreading} shows how as the ensemble size $R$ grows, the reconstruction error, which is formulated with respect to the whole Hamiltonian matrix, decreases faster and the fitted exponent $\beta$ approaches the theoretical single-parameter limit, predicted for multi-time stamp learning in the quadratic fisher information window. This convergence indicates that oscillatory interference terms in the measurement dynamics are progressively cancelled in the ensemble average. As a result, correlations between different Hamiltonian parameters are reduced and the Fisher information matrix becomes increasingly diagonal. This enables all Hamiltonian parameters to be recovered simultaneously from the same dataset. 

\begin{figure}[ht!]
  \centering
  \begin{subfigure}[t]{0.45\textwidth}
    \includegraphics[width=\textwidth]{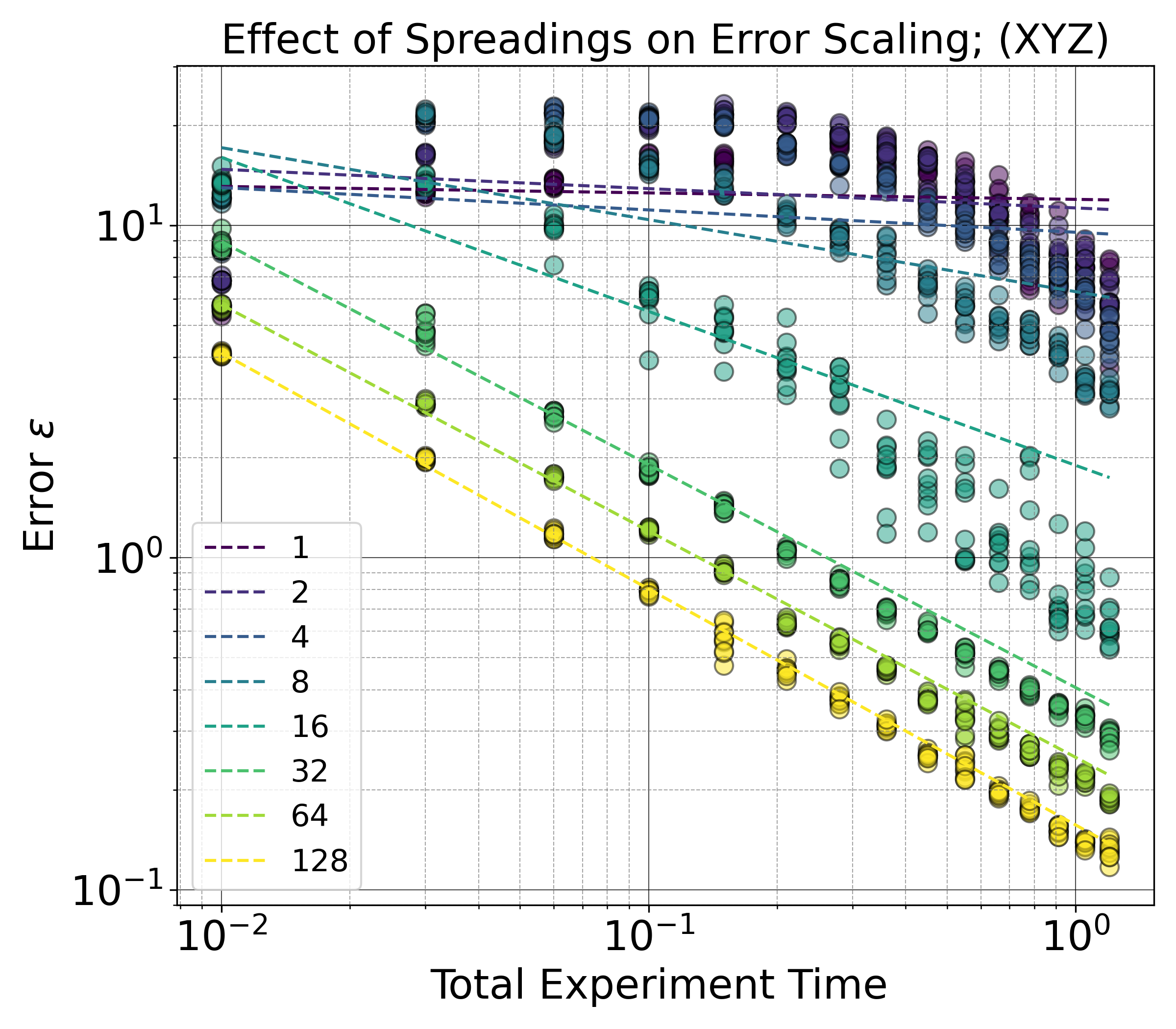}
    \caption{Reconstruction error $\varepsilon$ vs.\ $T_{\rm tot}$ for increasing Number of Spread States $R$.}
    \label{fig:perturbation_scaling}
  \end{subfigure}
  \hfill
  \begin{subfigure}[t]{0.45\textwidth}
    \includegraphics[width=\textwidth]{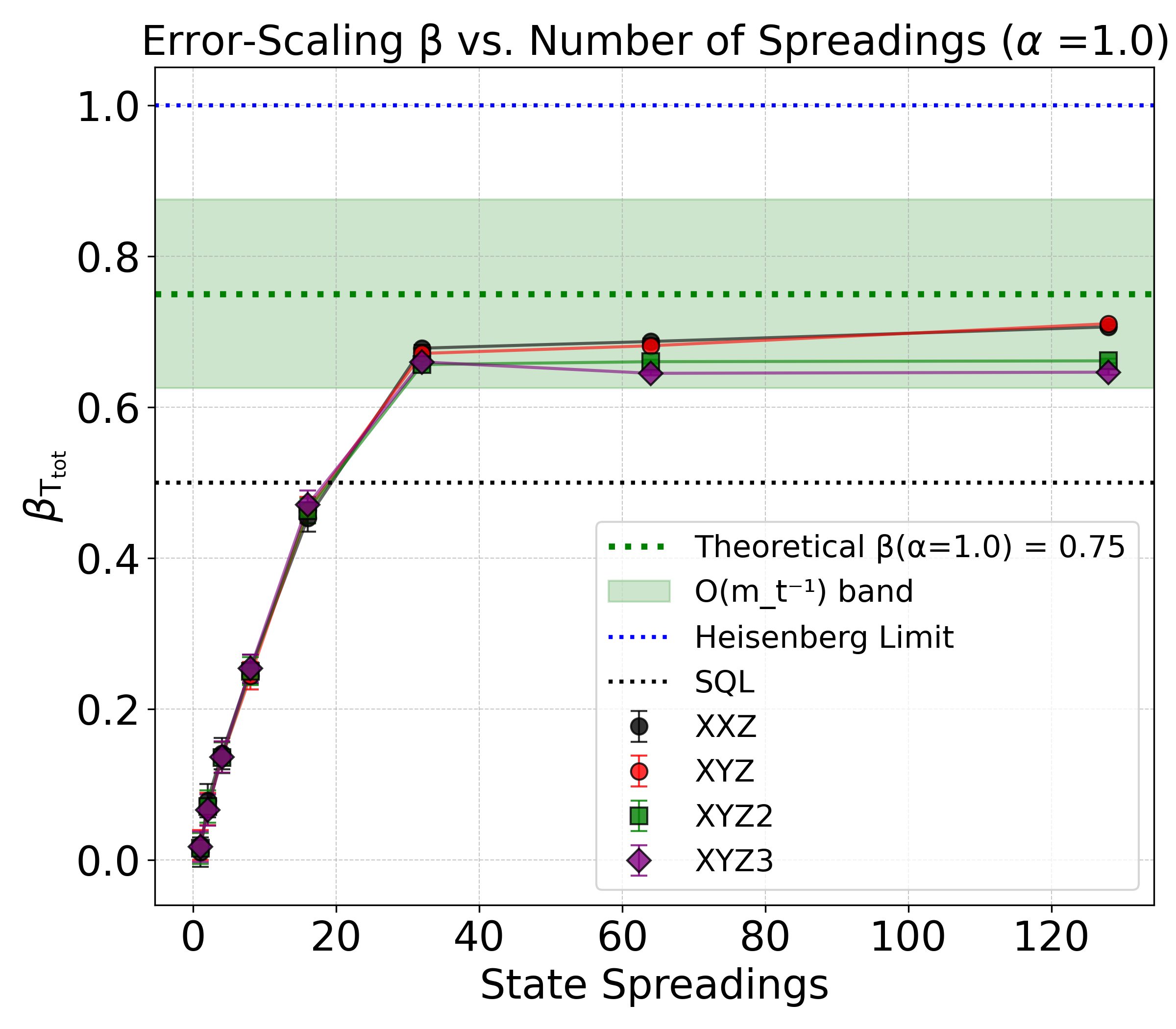}
    \caption{Extracted scaling exponent $\beta$ vs. Number of Spread States\ $R$.}
    \label{fig:state_spread_lr}
  \end{subfigure}
  \caption{
(a) Reconstruction error $\varepsilon$ versus total experiment time $T_{\rm tot}$ (per spread state) for growing number of spread states $R$ (XYZ model  Eq.~\ref{eq:first_test_ham}). 
(b) Corresponding scaling exponents $\beta$ obtained from $\varepsilon \propto T_{\rm tot}^{-\beta}$ across Hamiltonian families defined in Sec.~\ref{sec:model}. 
As $R$ increases, $\beta$ converges toward the prediction $\beta \approx 0.75$ (see Eqs.~\ref{eq:err_1} and \ref{eq:err_2} with $\alpha=1.0$, assuming Heisenberg-limited scaling with single evolution time $t$, $\gamma_0=2$). This behaviour indicates that averaging over spread states progressively suppresses parameter cross-correlations so that the Fisher information matrix becomes effectively diagonal (Theorem~\ref{thm:fisher_diagonalization}). As a result, the quadratic single-parameter sensitivity regime becomes simultaneously accessible for all Hamiltonian parameters, allowing the entire Hamiltonian to be recovered from the same dataset. See Appendix~\ref{app:spread_scaling_data} for the visualized data.}
  \label{fig:scaling_state_spreading}
\end{figure}

\subsection{Managing the Super-Linear learning rate Degradation in Multi-Timestamp Scenarios}
\label{sec:numerical_emergent_coherence}

Given that the Fisher information scales super-linearly with evolution time $t$, then it not necessarily, with total experiment time $T_{\rm tot} = \sum_k t_k$ at the same scaling rate ( since $(\sum_k t_k)^2 \neq \sum_k t_k^2$ ). As discussed in Sec.~\ref{sec:quadratic_scheduling}, this issue can be managed by non-uniform measurement scheduling, which we demonstrate here explicitly and illustrate with Figure.~\ref{fig:perturbation_scaling}.

In Figure~\ref{fig:scaling_alphas}, we see that increasing $\alpha$  increases the error scaling $\beta$ with total experiment time $T_{\rm tot} = \sum_k t_k$ as predicted by the theoretical curve, given in Eq.~\ref{eq:err_1}
derived for given Heisenberg-limited scaling with evolution time $t$.
\begin{figure}[h!]
  \centering
  \includegraphics[width=0.45\textwidth]{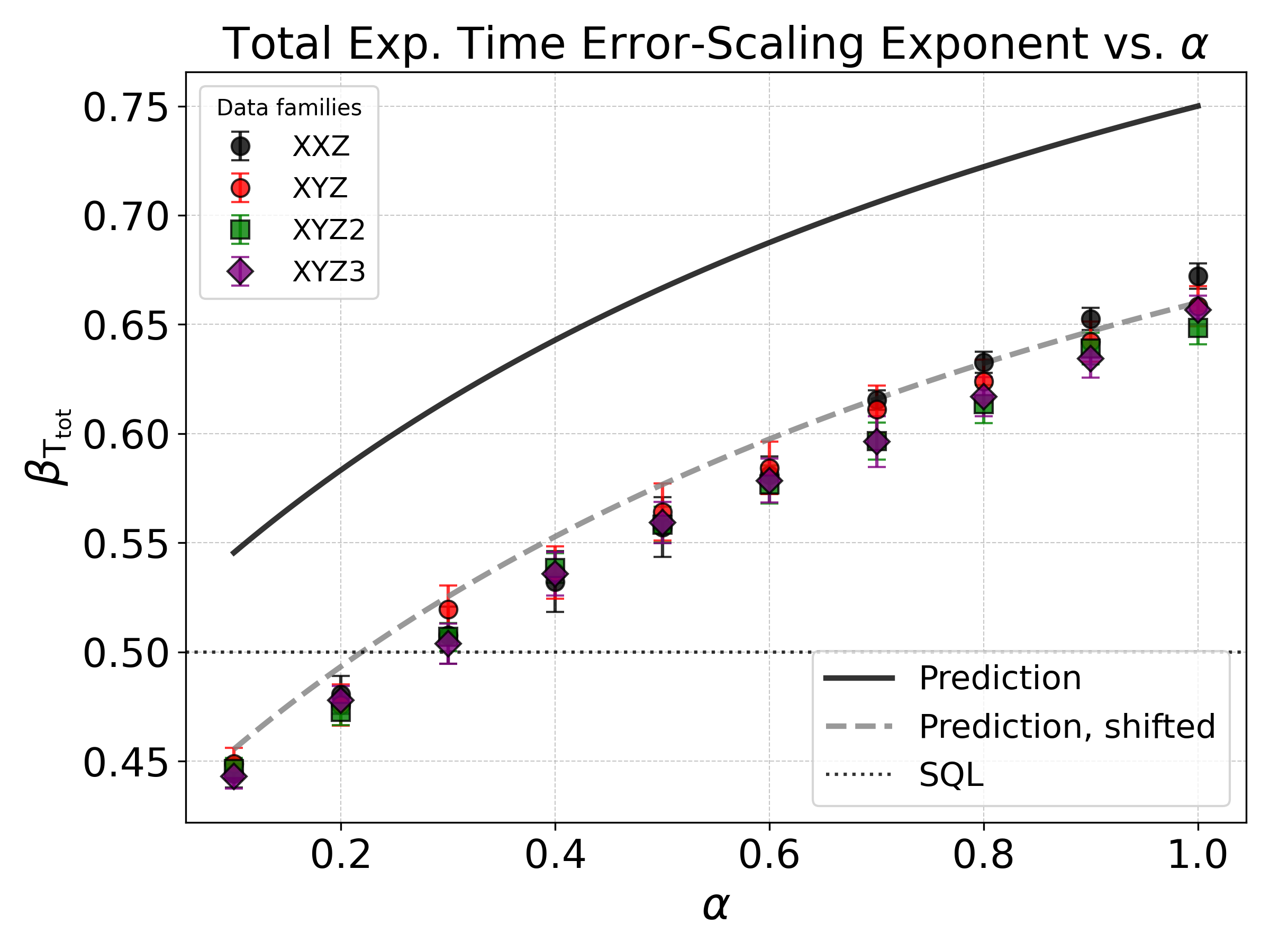}
  \caption{
Empirical scaling exponent $\beta_{T_{\rm tot}}$ as a function of the scheduling parameter $\alpha$.
The solid curve shows the theoretical prediction 
$\beta_{T_{\rm tot}}(\alpha) = \frac{1}{2}\,\frac{\alpha \gamma_0 + 1}{\alpha + 1}$ given based on assuming Heisenberg limited scaling:$\gamma_0 = 2$, and the dashed curve includes a small vertical offset accounting for finite ensemble and time stamp sampling effects on the diagonalization of the Fisher Matrix, affecting the recovery performance, that vanish asymptotically (see Props.~\ref{prop:cumulative_scaling} and Thm.~\ref{thm:fisher_diagonalization}). See Appendix~\ref{app:alpha_scaling_data} for visualized data.
}
  \label{fig:scaling_alphas}
\end{figure}
To examine this effect, we learn Hamiltonians from each Hamiltonian Model (given in Sec.~\ref{sec:model}), with $n=$5 qubits, for $\alpha \in \{0.1,0.2,\ldots,1.0\}$, each with an ensemble of a fixed number of $R=$ 32 spread states, with one-shot measurements in 25 random Pauli-product bases at $m_t = 8$ evolution times $t_k = \Delta t\,k^\alpha$ with $\Delta t = 0.01$.
For each scheduling parametrized by $\alpha$, we extracted the empirical error scaling $\varepsilon \propto T_{\rm tot}^{-\beta(\alpha)}$ as explained in Sec.~\ref{sec:recovery_error_and_total}.
The theoretically predicted learning rate behaviour with $\alpha$ is given at the asymptotic limit, meaning, assuming Heisenberg limited scaling with single evolution time $t$: $\gamma_0=2$ in Eqs.~ \ref{eq:err_1} and \ref{eq:err_2}.
The vertical offsets to the theoretically predicted behaviour in Figure~\ref{fig:scaling_alphas} remain constant across the range of $\alpha$. The existence of an offset agrees with the predicted asymptotic approach to full Fisher matrix diagonalization with a finite number of spread states (see Thm.~\ref{thm:fisher_diagonalization}) affecting the performance of the Hamiltonian matrix recovery convergence, and with the finite time stamp sampling effect (see Props.~\ref{prop:cumulative_scaling}).

\subsection{Fisher Information Diagnostic}
\label{sec:fisher_diagnostics_theory}

The recovery-error experiments above show that increasing the number of spread states enables the simultaneous learning of all Hamiltonian parameters at the theoretically predicted rate. This behaviour is consistent with the diagonalization of the Fisher information matrix, which removes parameter correlations and allows the quadratic single-parameter sensitivity regime to become accessible for all parameters simultaneously. This raises a more refined question related to the discussion at the end of Sec.~\ref{sec:random_fisher_scaling}: does the quadratic short-time sensitivity regime already practically exist when using only a single spread state together with sufficiently many measurement bases, or does it only emerge once a sufficient number of spread states is used, similar to the requirements for Fisher diagonalization?

To address this question we analyse the Fisher information of the measurement
statistics directly rather than relying on the recovery error.

Specifically, we consider the classical multi-parameter Fisher information
matrix \cite{liu2020quantum}. For a fixed measurement configuration defined by
a prepared spread state, measurement basis, and evolution time
$(r,j,k)$, the Fisher information matrix takes the form
\begin{equation}
\mathcal I^{(rjk)}_{ij}(\theta)
=
\sum_x
\frac{1}{p_x^{(rjk)}(\theta)}
\frac{\partial p_x^{(rjk)}(\theta)}{\partial \theta_i}
\frac{\partial p_x^{(rjk)}(\theta)}{\partial \theta_j},
\end{equation}
where $x$ denotes a measurement outcome. Because measurements performed with different configurations are statistically independent, the total Fisher information accumulated by the protocol is obtained by summing the contributions of all configurations,
\begin{equation}
\mathcal I_{ij}(\theta)
=
\sum_{r,j,k}
\sum_x
\frac{1}{p_x^{(rjk)}(\theta)}
\frac{\partial p_x^{(rjk)}(\theta)}{\partial \theta_i}
\frac{\partial p_x^{(rjk)}(\theta)}{\partial \theta_j}.
\end{equation}

The Fisher information quantifies the sensitivity of the measurement probabilities to variations of the Hamiltonian parameters. Each measurement configuration (defined by the prepared spread state, measurement basis, and evolution time) contributes an additive term to the total Fisher information.
Consequently, as additional measurement configurations are included in the dataset, the total Fisher information accumulates accordingly. This diagnostic therefore evaluates the information growth in the measurement statistics before any Hamiltonian reconstruction is performed.

We evaluate the trace of the Fisher matrix,
\begin{equation}
\mathrm{Tr}\,\mathcal I
=
\sum_i \mathcal I_{ii}
=
\sum_{r,j,k}\sum_x
\frac{\sum_i
\left(\partial_{\theta_i} p_x^{(rjk)}\right)^2}
{p_x^{(rjk)}}
=
\sum_{r,j,k}\sum_x
\frac{\|\nabla_\theta p_x^{(rjk)}\|^2}
{p_x^{(rjk)}}.
\end{equation}

The trace equals the sum of the eigenvalues of the Fisher matrix and therefore quantifies the total information contained in the measurement statistics about all Hamiltonian parameters, in a multi-parameter scenario. Because it can be evaluated without diagonalizing the Fisher matrix, it provides a direct probe of how the total Fisher information grows with evolution time in the case of multi-parameter Hamiltonians, even when parameter correlations obscure the behaviour of individual matrix elements (see Sec.~\ref{sec:fisher_isotropy}).

Here $p_x^{(rjk)}(\theta)$ denotes the probability of measurement outcome $x$ for Hamiltonian parameters $\theta$ in configuration $(r,j,k)$. In the experimental protocol considered in this work, these probabilities arise from preparing a spread state $\ket{\psi_0^{(r)}}$, evolving it under the true Hamiltonian $H(\theta)$ for time $t_k$, and measuring in a Pauli-product basis $p_j$, so that
\begin{equation}
p_x^{(rjk)}(\theta)
=
\left|
\langle x \mid p_j\,
e^{-i H(\theta) t_k}
\mid \psi_0^{(r)}
\rangle
\right|^2 .
\end{equation}

The parameter derivatives entering the Fisher information quantify the sensitivity of these probabilities to infinitesimal perturbations of the Hamiltonian parameters,
\begin{equation}
\partial_{\theta_i} p_x^{(rjk)}(\theta)
=
\lim_{\epsilon \to 0}
\frac{p_x^{(rjk)}(\theta+\epsilon e_i)-p_x^{(rjk)}(\theta)}{\epsilon}.
\end{equation}
where $e_i$ denotes the unit vector in parameter space. In the numerical experiments the probabilities are obtained from the simulated measurement model underlying the Hamiltonian estimator introduced in Sec.~\ref{sec:estimator}, and the gradients $\nabla_\theta p_x^{(rjk)}(\theta)$ are evaluated using automatic differentiation, which computes the derivatives up to machine precision.

In the multi-time measurement protocol used throughout this work, measurements are performed at evolution times \(t_k = \Delta t\,k^\alpha\). The cumulative Fisher information accumulated over the full trajectory therefore follows the scaling predicted in Proposition~\eqref{prop:cumulative_scaling},
\begin{equation}
I_{\rm tot}=\sum_k \mathcal I(t_k)
\propto T_{\rm tot}^{\,p},
\qquad
p=\frac{\alpha\gamma_0+1}{\alpha+1},
\label{eq:cumulative_fisher_scaling}
\end{equation}
where \(T_{\rm tot}=\sum_k t_k\) denotes the total experiment time. For the schedule used in the experiment (\(\alpha=1\) and \(\gamma_0=2\)), this yields the prediction \(p=\frac{3}{2}\),
\begin{equation}
I_{\rm tot}\propto T_{\rm tot}^{3/2}.
\end{equation}
To test this prediction we evaluate the Fisher trace as a function of the total experiment time and fit the resulting trajectories to a power law,
\begin{equation}
\mathrm{Tr}\,\mathcal I(T_{\rm tot}) = A\,T_{\rm tot}^{p}.
\end{equation}

\subsubsection{Quadratic Fisher-Scaling Diagnostics}
\label{sec:fisher_diagnostics_experiments}

In the diagnostic experiment we generate ensembles of spread states with sizes \(R \in \{1,2,4,8,16\}\) according to Def.~\ref{dfn:state_spreading}. Each spread state is measured in \(25\) random Pauli-product bases and evolved under the Hamiltonian at multiple evolution times following the schedule \(t_k = \Delta t\,k^\alpha\), with \(\Delta t=0.01\), fixed \(\alpha=1.0\), and \(k \in \{2,4,6,8,10\}\). For each measurement configuration we simulate the time evolution and compute the corresponding measurement probabilities. The Fisher information is then evaluated numerically using automatic differentiation of the probabilities with respect to the Hamiltonian parameters.

Figure~\ref{fig:fisher_diagnostic}(a) shows the scaling of the Fisher trace as a function of the total experiment time \(T_{\rm tot}\) for different spread-state ensemble sizes \(R\) for a representative Hamiltonian family (here the XYZ model). The diagnostic experiment uses a system size of \(n=5\) qubits and evaluates time stamps \(t_k=\Delta t\,k^\alpha\) with \(\Delta t=0.01\), \(\alpha=1.0\), and \(k\in\{2,4,6,8,10\}\). For each configuration we generate spread-state ensembles of size \(R\in\{1,2,4,8,16\}\) and measure each state in \(25\) one-shot random Pauli-product bases. The trajectories follow a clear power-law behaviour consistent with \(\mathrm{Tr}\,\mathcal I(T_{\rm tot}) \propto T_{\rm tot}^{3/2}\), confirming the predicted cumulative Fisher scaling for the multi-time stamp protocol scaling \(p=\frac3 2\) for time stamp scheduling with \(\alpha=1.0\), visible in Eq.~\eqref{eq:cumulative_fisher_scaling}.

Figure~\ref{fig:fisher_diagnostic}(b) shows the corresponding fitted scaling exponent \(p\) as a function of \(R\), aggregated across the Hamiltonian families defined in Sec.~\ref{sec:model}. For this experiment we keep the system size fixed at \(n=5\) qubits and use the same measurement configuration and time schedule as in panel (a). Across all tested ensemble sizes \(R\), the numerical experiment yields exponents \(p\approx1.49\), which is in agreement with the theoretical prediction \(p=3/2\). Notably, the predicted scaling already appears for the case \(R=1\), demonstrating that the quadratic short-time sensitivity regime established in Sec.~\ref{sec:random_fisher_scaling} is present even with a single spread state when averaging over sufficiently many measurement bases. Increasing the number of spread states therefore primarily serves to diagonalize the Fisher information matrix and improve the conditioning of the multiparameter estimation problem, rather than creating the quadratic regime itself.

\begin{figure}[ht!]
  \centering
  \begin{subfigure}[t]{0.45\textwidth}
    \includegraphics[width=\textwidth]{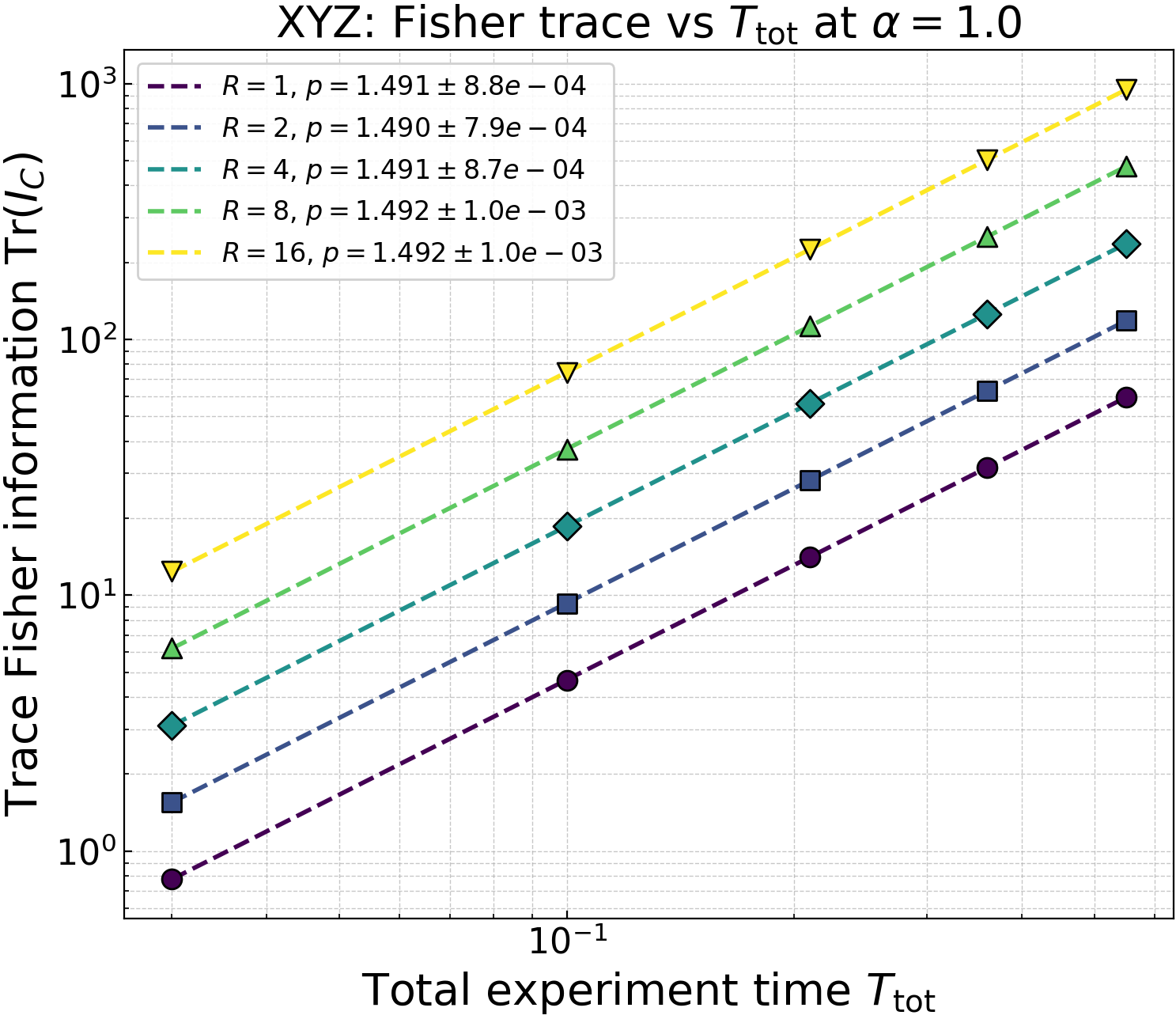}
    \caption{Trace of the Fisher information $\mathrm{Tr}\,\mathcal I$ versus total experiment time $T_{\rm tot}$ for different spread-state ensemble sizes $R$ (XYZ model).}
    \label{fig:fisher_trace_scaling}
  \end{subfigure}
  \hfill
  \begin{subfigure}[t]{0.45\textwidth}
    \includegraphics[width=\textwidth]{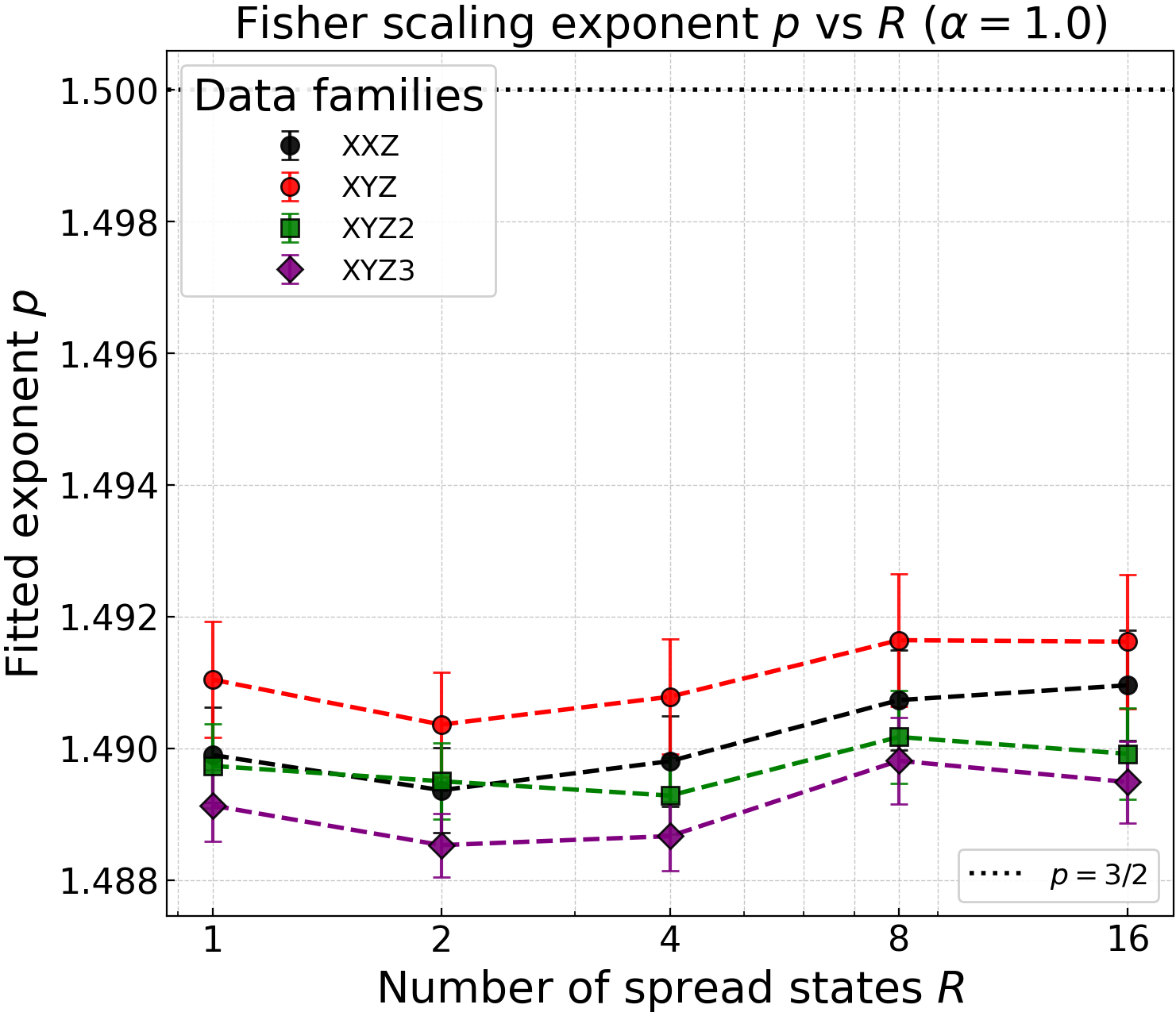}
    \caption{Fitted cumulative Fisher-scaling exponent $p$ as a function of the number of spread states $R$.}
    \label{fig:fisher_exponent_vs_R}
  \end{subfigure}
\caption{
Fisher-information diagnostic for varying numbers of spread states.
(a) Growth of the Fisher trace $\mathrm{Tr}\,\mathcal I$ with total experiment time $T_{\rm tot}$ for different ensemble sizes $R$. The trajectories follow a power-law behaviour consistent with $\mathrm{Tr}\,\mathcal I \propto T_{\rm tot}^{3/2}$ for the time-stamp schedule used in the experiment.
(b) Corresponding fitted scaling exponent $p$ as a function of $R$, aggregated across the Hamiltonian families defined in Sec.~\ref{sec:model}. The fitted exponents remain close to the theoretical prediction $p=3/2$ even for $R=1$. This demonstrates that the quadratic short-time sensitivity regime is already present when using a single spread state together with sufficiently many measurement bases. Increasing the number of spread states therefore primarily improves the conditioning of the multiparameter estimation problem by diagonalizing the Fisher matrix rather than creating the quadratic regime itself.
See Appendix~\ref{app:fisher_trace_data} and Appendix~\ref{app:fisher_exponent_spreadings} for the visualized data.
}
\label{fig:fisher_diagnostic}
\end{figure}

Finally, as Eq.~\eqref{eq:cumulative_fisher_scaling} predicts, the cumulative Fisher exponent \(p\) should vary with the scheduling parameter \(\alpha\). To validate this prediction directly on the Fisher information, we analyse the case of a single spread state (\(R=1\)) while varying the scheduling parameter \(\alpha\). The experiment is repeated across the Hamiltonian families defined in Sec.~\ref{sec:model}. For each value of \(\alpha\) we evaluate the Fisher trace as a function of the total experiment time and extract the corresponding scaling exponent \(p\). For this experiment we again consider systems with \(n=5\) qubits and evaluate time stamps \(t_k=\Delta t\,k^\alpha\) with \(\Delta t=0.01\) and \(k\in\{2,4,6,8,10\}\), while varying the scheduling exponent \(\alpha\in\{0.2,0.4,0.6,0.8,1.0\}\) and keeping the spread-state ensemble fixed at \(R=1\).

The dependence \(p(\alpha)\) is validated in Fig.~\ref{fig:fisher_alpha_diagnostic}. The extracted exponents closely follow the theoretical prediction \(p=(\alpha\gamma_0+1)/(\alpha+1)\), demonstrating that the cumulative Fisher scaling is governed by the measurement-time scheduling and confirming that the previously observed scaling \(p\approx3/2\) for the schedule \(\alpha=1\) arises from the general relation between the single-experiment Fisher scaling and the scheduling protocol described by Eq.~\eqref{eq:cumulative_fisher_scaling}, and was not a mere coincidence.

\begin{figure}[ht!]
  \centering
  \includegraphics[width=0.4\textwidth]{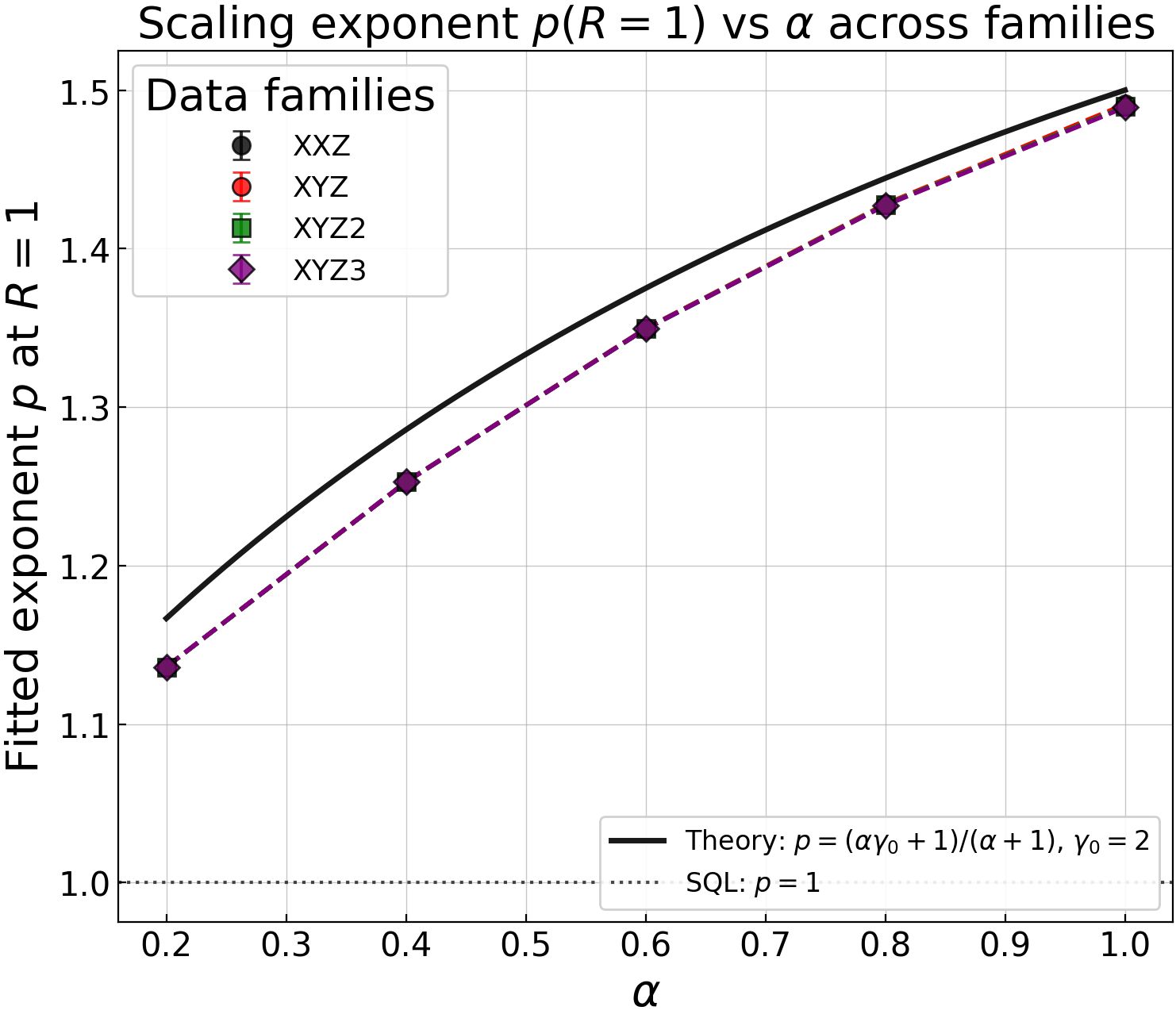}
\caption{Fitted cumulative Fisher-information scaling exponent \(p\) as a function of the time stamp scheduling parameter \(\alpha\) for a single spread state (\(R=1\)). Data points correspond to numerical fits of $\mathrm{Tr}\,\mathcal I(T_{\rm tot}) \propto T_{\rm tot}^{p}$, while the solid line shows the theoretical prediction $p=(\alpha\gamma_0+1)/(\alpha+1)$ with $\gamma_0=2$. This demonstrates that the cumulative Fisher scaling is indeed governed by the measurement-time scheduling. See Appendix~\ref{app:fisher_exponent_alpha} for the visualized data.}
  \label{fig:fisher_alpha_diagnostic}
\end{figure}

Importantly, we also verify that the spectral-activation mechanism described in Sec.~\ref{sec:random_fisher_scaling} remains effective as the system size increases. Because the spread-state construction applies independent Haar-random rotations to each qubit, the resulting randomization acts locally and therefore does not deteriorate with increasing system size. To demonstrate this behaviour numerically, we repeat the Fisher diagnostic for systems ranging from $n=2$ to $n=6$ qubits while keeping $R=1$ and $\alpha=1$. The experiment uses the same measurement configuration as in the previous Fisher diagnostics, with $5$ one-shot random Pauli-product measurement bases and time stamps $t_k=\Delta t\,k^\alpha$ with $\Delta t=0.01$ and $k\in\{2,4,6,8,10\}$.
\begin{figure}[t]
\begin{minipage}{0.48\linewidth}
\centering
\includegraphics[width=\linewidth]{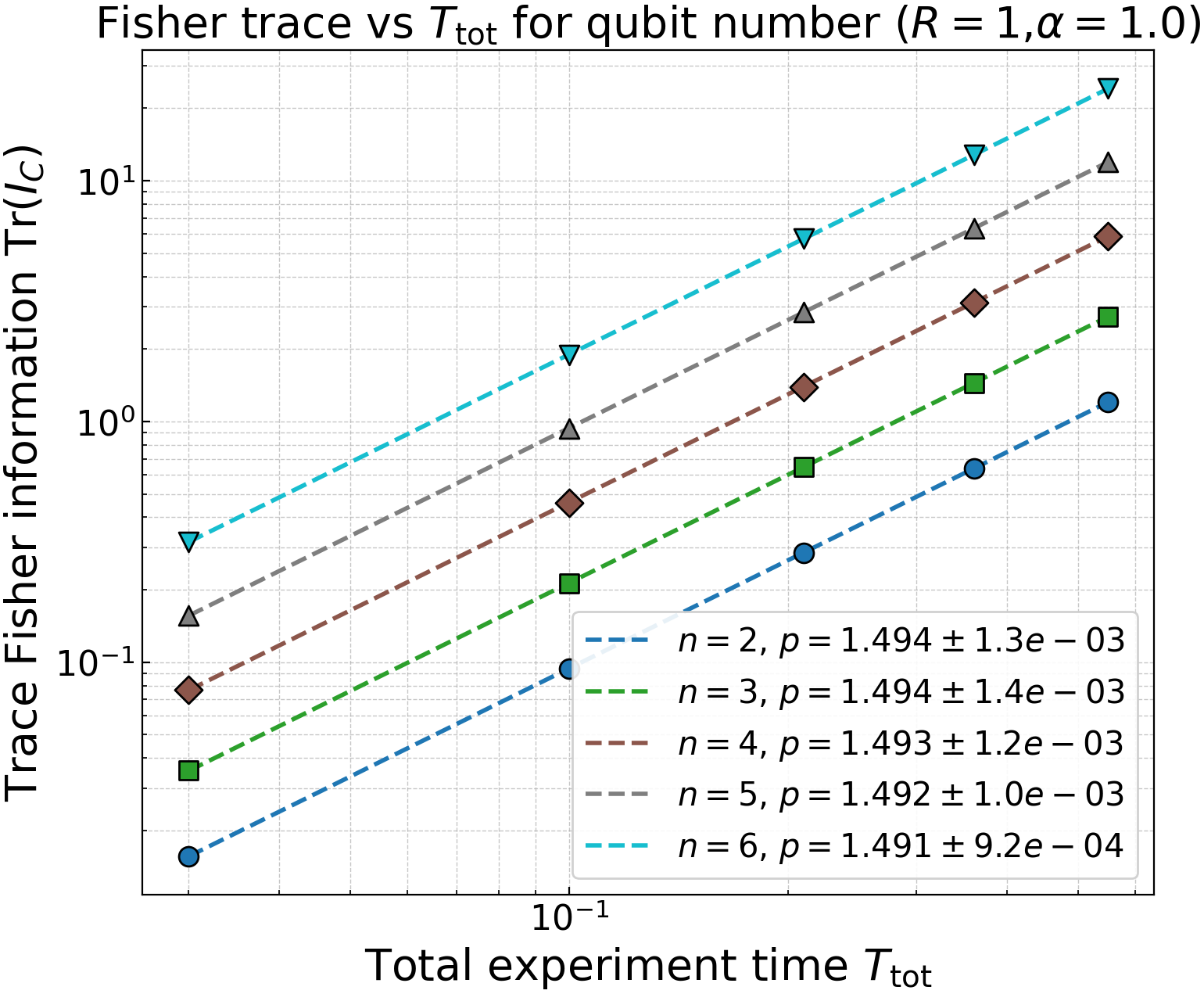}
\caption{
Scaling of the trace of the classical Fisher information $\mathrm{Tr}\,\mathcal I$ as a function of the total experiment time $T_{\rm tot}$ for different system sizes. 
Each curve corresponds to a different number of qubits $n$ with fixed spread-state ensemble size $R=1$ and scheduling exponent $\alpha=1$. 
Dashed lines show power-law fits $\mathrm{Tr}\,\mathcal I \propto T_{\rm tot}^{p}$. 
Across the range of qubit numbers $n=2$ to $n=6$, the fitted exponents remain close to the theoretical prediction $p=\tfrac{3}{2}$ for multi time stamp scenarios (as given in Sec.~\ref{sec:quadratic_scheduling}), indicating that the underlying quadratic short-time Fisher scaling for single evolution times persists, independent of the system size.  The employed Hamiltonian is from the XYZ family (see Sec.~\ref{sec:model}). See Appendix~\ref{app:fishertrace_qubits} for the visualized data.
}
\label{fig:fisher_qubit_scaling}
\end{minipage}
\hfill
\begin{minipage}{0.48\linewidth}
\centering
\includegraphics[width=\linewidth]{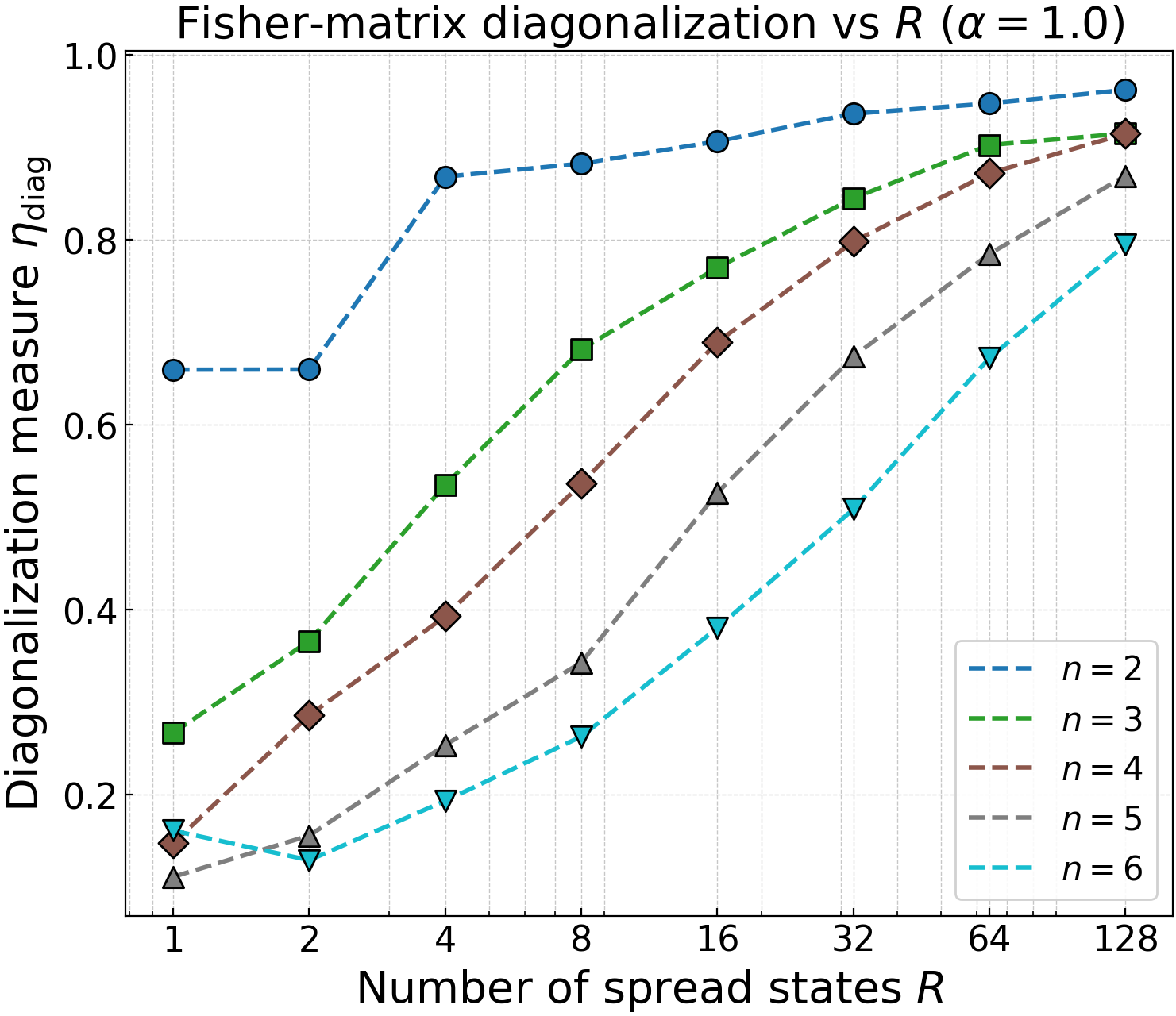}
\caption{
Diagonalization of the Fisher information matrix as a function of the number of spread states \(R\). The diagonalization measure \(\eta_{\mathrm{diag}}\) quantifies the relative weight of the diagonal entries of the Fisher matrix compared to its full Frobenius norm.  Values close to \(1\) indicate that the Fisher information is largely concentrated on the diagonal, corresponding to weak correlations between Hamiltonian parameters. Across the range of qubit numbers $n=2$ to $n=6$, the numerical results show that increasing the spread-state ensemble progressively suppresses off-diagonal Fisher couplings.  While larger systems exhibit stronger initial parameter correlations, the same diagonalization trend with increasing \(R\) persists across system sizes.  The employed Hamiltonian is from the XYZ family (see Sec.~\ref{sec:model}). See Appendix~\ref{app:fisher_diag_R} for the visualized data.
}
\label{fig:fisher_diag_eta}
\end{minipage}
\end{figure}
As shown in Fig.~\ref{fig:fisher_qubit_scaling}, the fitted Fisher-scaling exponent remains close to the theoretical value $p=\tfrac{3}{2}$ across all tested system sizes, indicating that the quadratic short-time regime remains practically accessible as the system grows.

These diagnostics therefore clarify the origin of the scaling behaviour observed in the recovery experiments. The quadratic short-time Fisher-information regime predicted in Sec.~\ref{sec:random_fisher_scaling} is already present when using a single spread state together with sufficiently many measurement bases. Increasing the number of spread states does not create this regime but instead improves the conditioning of the multiparameter estimation problem by suppressing parameter correlations and making the Fisher information matrix progressively closer to diagonal. The Fisher diagnostics thus confirm that the enhanced learning performance observed in the recovery experiments originates from the intrinsic information growth of the measurement statistics rather than from artefacts of the reconstruction procedure.

\subsubsection{Fisher-Matrix Diagonalization Diagnostics}
\label{sec:fisher_diagnostics_diagon}
While the Fisher-scaling diagnostics establish that the measurement protocol generates the predicted quadratic short-time information growth, they do not yet validate that increasing the number of spread states causes the stability of the Hamiltonian reconstruction observed in the recovery experiments. In particular, the recovery task is a multiparameter estimation problem in which correlations between different Hamiltonian parameters can degrade the conditioning of the inference problem, making isolation of each parameter necessary for its inference (see Sec.\ref{sec:fisher_isotropy}), even when the total Fisher information grows favourably. We have observed in the numerical Hamiltonian recovery experiment in Sec.~\ref{sec:numerical_emergent_coherence}, that increasing the spread state ensemble size improves the recoverability of all Hamiltonian parameters at once from the same dataset. To validate that this effect directly is inherited from the resulting Fisher information matrix structure (as predicted in theory Sec.\ref{sec:fisher_isotropy}), we therefore analyse the structure of the classical Fisher information matrix itself rather than only its trace.

The spread-state construction suggests a natural mechanism for suppressing such parameter correlations. Because each spread state applies independent Haar-random single-qubit rotations before measurement, different Hamiltonian terms are probed under locally randomized orientations. As shown analytically in Theorem~\ref{thm:fisher_diagonalization}, averaging over multiple spread states therefore progressively decorrelates the sensitivities of the measurement probabilities with respect to different Hamiltonian parameters. As a result, the Fisher information matrix is expected to become increasingly diagonal as the number of spread states \(R\) grows, improving the conditioning of the multiparameter estimation problem without changing the intrinsic Fisher scaling established in Thm.~\ref{thm:fisher_window}.

To quantify this diagonalization effect, we evaluate the relative weight of the diagonal part of the Fisher information matrix compared to its full matrix norm. Specifically, given the Fisher information matrix accumulated over all measurement configurations $(r,j,k)$,
\begin{equation}
\mathcal I_{ij}(\theta)
=
\sum_{r,j,k}
\sum_x
\frac{1}{p_x^{(rjk)}(\theta)}
\frac{\partial p_x^{(rjk)}(\theta)}{\partial \theta_i}
\frac{\partial p_x^{(rjk)}(\theta)}{\partial \theta_j},
\end{equation}
we define the diagonalization measure
\begin{equation}
\eta_{\mathrm{diag}}
=
\frac{\sum_i \mathcal I_{ii}^2}{\sum_{i,j} \mathcal I_{ij}^2}
=
\frac{\|\mathrm{diag}(\mathcal I)\|_F^2}{\|\mathcal I\|_F^2}.
\end{equation}
Here $\|\cdot\|_F$ denotes the Frobenius norm and $\mathrm{diag}(\mathcal I)$ denotes the matrix containing only the diagonal elements $\mathcal I_{ii}$. Values $\eta_{\mathrm{diag}}\approx1$ indicate that most of the Fisher information is concentrated in the diagonal entries, corresponding to weak parameter correlations, whereas smaller values signal stronger off-diagonal couplings between parameters. By analysing $\eta_{\mathrm{diag}}$ as a function of the spread-state ensemble size $R$, we directly quantify how increasing the number of spread states suppresses parameter correlations and improves the conditioning of the multi-parameter estimation problem.

We perform a Fisher diagnostic experiment in which we evaluate the Fisher information matrix for increasing spread-state ensemble sizes \(R \in \{1,2,4,8,16,32,64,128\}\). The experiment is repeated for systems ranging from \(n=2\) to \(n=6\) qubits while keeping the scheduling parameter fixed at \(\alpha=1\),  with $1$ one-shot random Pauli-product measurement bases. To focus on the instantaneous structure of the Fisher matrix, the diagnostic is evaluated at the early time stamps corresponding to \(k\in\{1\}\) in the time schedule \(t_k=\Delta t\,k^\alpha\) with \(\Delta t=0.01\).

Figure~\ref{fig:fisher_diag_eta} shows the resulting diagonalization measure \(\eta_{\mathrm{diag}}\) as a function of \(R\) for different system sizes. As predicted by the analytical result of Theorem~\ref{thm:fisher_diagonalization}, the Fisher matrix becomes progressively more diagonal as the spread-state ensemble increases. While larger systems exhibit a lower initial diagonalization level, reflecting stronger parameter correlations in the larger Hamiltonian parameter space, the same qualitative diagonalization trend with increasing \(R\) is observed across all tested system sizes. This indicates that the decorrelation mechanism induced by locally randomized spread states remains effective beyond small Hilbert-space dimensions, although achieving comparable levels of Fisher-matrix diagonalization may require moderately larger spread-state ensembles as the number of Hamiltonian parameters increases.

These diagnostics therefore shows why increasing the number of spread states improves the robustness of the Hamiltonian reconstruction observed in the recovery experiments. While the intrinsic Fisher scaling determines how rapidly information accumulates with total experiment time, the diagonalization of the Fisher matrix governs the conditioning of the multi-parameter estimation problem. Increasing the spread-state ensemble therefore enables simultaneous learning of multiple Hamiltonian parameters by suppressing parameter correlations while preserving the favourable Fisher scaling established in the previous section.

\section{Discussion and Outlook}
\label{sec:summary_discussion}

 We introduced a Hamiltonian learning strategy that makes the naturally occurring quadratic short-time Fisher-information regime practically usable for Hamiltonian learning. Rather than creating a new scaling law, the protocol enables the generic quadratic Fisher scaling to be exploited over a sufficiently long temporal window to become operationally useful. The key mechanism is the use of locally random spread states together with random Pauli-product measurements. Spread states activate the Hamiltonian eigenbasis and prevent state-induced suppression of spectral gaps, while randomized measurements reveal these spectral components in the measurement statistics. As a result, the quadratic Fisher-information regime can be accessed in practice without requiring entanglement, coherent joint measurements, or dynamical control, relying only on local state preparation, local measurements, and classical data aggregation.

Secondarily, we showed that ensemble averaging over spread states enables multi-parameter Hamiltonian learning without requiring structural isolation of parameters. As the ensemble size increases, statistical correlations between parameters are progressively suppressed and the Fisher information matrix becomes effectively diagonal. This Fisher diagonalization removes cross-parameter correlations and allows all Hamiltonian parameters to be estimated simultaneously from a single data set, without relying on structural priors such as sparsity, or commutativity.

To demonstrate the practical implications of these mechanisms, we performed numerical Hamiltonian recovery experiments across several disordered many-body Hamiltonians. These simulations show consistent beyond-SQL error scaling and parameter-independent learning rates using only local product states and local measurements, illustrating that the protocol enables efficient Hamiltonian learning in realistic settings.

To verify that this behaviour originates from the information content of the measurement statistics rather than from artefacts of the reconstruction procedure such as convergence, we additionally analyse the Fisher information directly along the recovery process. The Fisher diagnostics confirm that the measurement data indeed exhibit the theoretical Fisher-information behaviour predicted by the analysis, including both the generic spectral activation resulting in quadratic Fisher scaling and predicted cumulative scaling and, secondly, the diagonalization of the Fisher information matrix. In addition, we examine how these mechanisms behave as the system size increases, and validate their size independent persistence. This separation between information availability and reconstruction performance demonstrates that the improved learning behaviour is rooted in the intrinsic information structure of the measurement statistics rather than in particular properties of the optimization procedure.

Together, these results establish that beyond-SQL learning rates in Hamiltonian learning can be reached without entanglement or coherent control, using experimentally accessible tools compatible with near-term devices. This provides a practical path toward resource-efficient quantum characterization and adaptive metrology. Future work will focus on quantifying robustness to state-preparation and measurement errors, extending state-spreading techniques to time-dependent Hamiltonians, and integrating compressed-sensing or tensor-network approaches for larger systems.

\bibliographystyle{naturemag}  
\bibliography{references}

\section*{Acknowledgments}

We would like to thank Jos\'e Ramon-Martinez, Marcin P\l{}odzie\'{n}, Carlos Pascual, and Antonio Ac\`in for their helpful comments and questions whilst proof-reading the manuscript.

B. Baran acknowledges support from AIDAS-AI, Data Analytics
and Scalable Simulation, which is a Joint Virtual Laboratory gathering between the Forschungszentrum Jülich and the French Alternative Energies and Atomic Energy Commission (CEA).

T. Heightman acknowledges support from the Government of Spain (Severo Ochoa CEX2019-000910-S, Quantum in  Spain, FUNQIP and European Union NextGenerationEU PRTR- C17.I1), the European Union (PASQuanS2.1, 1011 13690 and Quantera Veriqtas), Fundació Cellex, Fundació Mir- Puig, Generalitat de Catalunya (CERCA program), the  ERC AdG CERQUTE and the AXA Chair in Quantum  Information Science.

The authors declare no competing financial interest.

\section*{Author Contributions}
B.B. conceived the initial solution and led the development of the theoretical framework. T.H. conceived the research question and contributed significantly to the maturation of the theory. Both authors jointly prepared and finalized the manuscript.

\appendix

\section{Bachmann-Landau Notation}
\label{app:landau_notation}
In this work we will make use of the Bachmann–Landau notation also known as the asymptotic notation, particularly the big O notation, the big Omega notation and the big Theta notation \cite{howell2008asymptotic}.

For example, consider a function \(f(t)\) and the exponent \(x\):
\begin{equation}
\begin{aligned}
f(t) &= O(t^x) : \exists\,C > 0,\; t_0,\;\; f(t) \le C\,t^x \quad \forall\,t \ge t_0,\\
f(t) &= \Omega(t^x) : \exists\,c > 0,\; t_0,\;\; f(t) \ge c\,t^x \quad \forall\,t \ge t_0,\\
f(t) &= \Theta(t^x) : f(t) = O(t^x) \;\text{and}\; f(t) = \Omega(t^x).
\end{aligned}
\end{equation}

\section{Spectral Origin of the Quadratic Fisher Window}
\label{app:spectral_fisher_window}
\begin{proof}

Expanding the time-evolved state and measurement basis in the Hamiltonian eigenbasis yields
\begin{align}
\ket{\psi_t} &= \sum_k a_k e^{-i\lambda_k t}\ket{\lambda_k}, \\
\ket{m_j} &= \sum_k c_{jk}\ket{\lambda_k}.
\end{align}

The measurement probability becomes
\begin{equation}
p_j(t)
=
\left|\langle m_j \mid \psi_t \rangle\right|^2
=
\sum_{k,\ell}
c_{jk}^* c_{j\ell} a_k a_\ell^*
e^{-i(\lambda_k-\lambda_\ell)t}.
\end{equation}

The dynamics therefore consists of sinusoidal/harmonic oscillations with frequencies given by the spectral gaps \begin{equation}
\Delta_{k\ell} = \lambda_k - \lambda_\ell.
\end{equation}
Large spectral gaps generate rapid oscillations, while small gaps determine the longest oscillation periods.A spectral-gap contribution $(k,\ell)$ will be called \emph{activated} if its coefficient $c_{jk}^* c_{j\ell} a_k a_\ell^*$ is nonzero, so that the corresponding oscillatory term contributes to $p_j(t)$.

Differentiating with respect to the Hamiltonian parameter gives
\begin{equation}
\partial_\theta p_j(t)
=
T_{\mathrm{vec}}(t)+T_{\mathrm{val}}(t),
\end{equation}
where
\begin{equation}
T_{\mathrm{vec}}(t)
=
\sum_{k,\ell}
\partial_\theta(c_{jk}^* c_{j\ell} a_k a_\ell^*)
e^{-i(\lambda_k-\lambda_\ell)t},
\end{equation}
and
\begin{equation}
T_{\mathrm{val}}(t)
=
-it
\sum_{k,\ell}
c_{jk}^* c_{j\ell} a_k a_\ell^*
e^{-i(\lambda_k-\lambda_\ell)t}\partial_\theta(\lambda_k-\lambda_\ell).
\end{equation}

The term $T_{\mathrm{vec}}(t)$ remains bounded in time, whereas
$T_{\mathrm{val}}(t)$ carries an explicit prefactor $t$. In particular,
$T_{\mathrm{val}}(t)$ is a sum of oscillatory gap contributions of the form
\(
c_{jk}^* c_{j\ell} a_k a_\ell^*
e^{-i(\lambda_k-\lambda_\ell)t}\partial_\theta(\lambda_k-\lambda_\ell).
\)
For times
\begin{equation}
t \le \frac{\pi}{2\,\Delta\lambda_{\max}^{\mathrm{act}}},
\end{equation}
all activated oscillatory factors remain in their pre-extremal regime, so
their phase-dependent contributions vary smoothly. 

For spread states with Haar-random coefficients, a spectral gap fails to be activated happens with measure-zero due to the nature of the Haar distribution, see (Lemma~\ref{lem:generic_spectral_activation}). Thus the oscillatory sum contains all gap modes.

Persistent destructive interference among these activated oscillatory terms would require fine-tuned relations among the coefficients and spectral gaps, but since the coefficients inherit a continuous Haar-random distribution from the spread state, such relations occur only on a measure-zero subset of
coefficient space. Consequently the oscillatory sum
\[
F_j(t)=
-i\sum_{k,\ell}
c_{jk}^*c_{j\ell}a_ka_\ell^*
e^{-i\Delta_{k\ell}t}\partial_\theta\Delta_{k\ell}
\]
remains generically bounded away from zero on this pre-extremal interval, i.e.\
\(F_j(t)=\Omega(1)\). Hence \(T_{\mathrm{val}}(t)=\Omega(t)\).

Since $T_{\mathrm{vec}}(t)=O(1)$ remains bounded while
$T_{\mathrm{val}}(t)=\Omega(t)$ grows linearly in $t$, the eigenvalue
contribution dominates on this window, and therefore
\begin{equation}
\partial_\theta p_j(t)=\Omega(t).
\end{equation}

Substituting into the classical Fisher information
\begin{equation}
\mathcal I_C(t)
=
\sum_j
\frac{(\partial_\theta p_j(t))^2}{p_j(t)},
\end{equation}
and using that \(0<p_j(t)\le 1\), so that
\begin{equation}
\frac{(\partial_\theta p_j(t))^2}{p_j(t)}
\ge
(\partial_\theta p_j(t))^2,
\end{equation}
therefore yields, generically,
\begin{equation}
\mathcal I_C(t)=\Omega(t^2)
\end{equation}
on the same pre-extremal temporal window on which
\(\partial_\theta p_j(t)=\Omega(t)\), namely
\begin{equation}
t \in
\left[
0,
\frac{\pi}{2\,\Delta\lambda_{\max}^{\mathrm{act}}}
\right].
\end{equation}

\end{proof}

\section{Lemma: Generic Spectral Activation }
\label{lem:generic_spectral_activation}
Let the spread state be
\(
\ket{\psi_{\rm spread}}
=
\bigotimes_{r=1}^n U_r \ket{\phi_r},
\)
where each \(U_r\sim\mathrm{Haar}(\mathrm{SU}(2))\) is sampled independently and \(\ket{\phi_r}\) are fixed single-qubit states. Let \(H\ket{\lambda_k}=\lambda_k\ket{\lambda_k}\) be the spectral decomposition of an \(n\)-qubit Hamiltonian. Then the spread state can be expanded in the Hamiltonian eigenbasis as
\(
\ket{\psi_{\rm spread}}
=
\sum_k a_k \ket{\lambda_k},
\,\,
a_k=\langle\lambda_k|\psi_{\rm spread}\rangle .
\)
Then, with overwhelming probability, over the choice of the local Haar-random unitaries \(\{U_r\}\), the amplitudes \(a_k\) are nonzero for all \(k\).

\begin{proof}
Let
\(
\ket{\psi_{\rm spread}}
=
\bigotimes_{j=1}^n U_j\ket{\phi_j},
\)
where each \(U_j\sim \mathrm{Haar}(\mathrm{SU}(2))\) is sampled independently and the states \(\ket{\phi_j}\) are fixed single-qubit states.

For any fixed Hamiltonian eigenstate \(\ket{\lambda_k}\), define the overlap
\begin{equation}
a_k
=
\langle \lambda_k | \psi_{\rm spread} \rangle
=
\left\langle \lambda_k \middle| \bigotimes_{j=1}^n U_j \middle| \bigotimes_{j=1}^n \phi_j \right\rangle
=: F_k(U_1,\dots,U_n).
\end{equation}

The function \(F_k\) is continuous in the sampled local unitaries. Moreover, \(F_k\) is not identically zero: varying the local unitaries \(\{U_j\}\) generates arbitrary product states, and product states span the full \(n\)-qubit Hilbert space. Therefore \(\ket{\lambda_k}\) cannot be
orthogonal to all such states, implying that \(F_k\not\equiv 0\).

Since \(\mathrm{SU}(2)\) is uncountable and equipped with a continuous Haar measure, the probability of sampling a tuple of local unitaries such that \(F_k(U_1,\dots,U_n)=0\) is zero, because the corresponding zero set is negligible relative to the full parameter space:
\begin{equation}
\mathbb{P}[a_k=0]
=
\mathbb{P}[F_k(U_1,\dots,U_n)=0]
=
0 .
\end{equation}

Finally, since there are only finitely many eigenstates \(\ket{\lambda_k}\), the union over all \(k\) of these probability-zero events still has total probability zero:
\begin{equation}
\mathbb{P}\!\left[\exists\,k:\,a_k=0\right]
\le
\sum_k \mathbb{P}[a_k=0]
=
0 .
\end{equation}

Hence, with probability one over the choice of the independently sampled local Haar-random unitaries, the amplitudes
\(
a_k=\langle\lambda_k|\psi_{\rm spread}\rangle
\)
are non-zero for all \(k\).
\end{proof}

\section{Lemma: Generic Sensitivity}
\label{lemma:generic_sensitivity}

Let \( \ket{\psi_0} = \bigotimes_{j=1}^n U_j \ket{\phi_j} \), where \( \ket{\phi} = \bigotimes_{j=1}^n \ket{\phi_j} \) is any fixed product eigenstate of a Pauli string Hamiltonian, and each \( U_j \sim \text{Haar}(\mathrm{SU}(2)) \) independently. Let \( \ket{b} = \bigotimes_{j=1}^n \ket{b_j} \) be a fixed computational basis state (or more generally, a Pauli product eigenstate). Then, with overwhelming probability, the choice of the \( \{U_j\} \), the overlaps \( \langle b | \psi_0 \rangle \) and \( \langle b | P_a | \psi_0 \rangle \) are non-zero for all \( a \), so that no term in the Hamiltonian is missed at first order in time, meaning that the measurement probability is sensitive to every Hamiltonian parameter.
\begin{proof}
Let \( \ket{\psi_0} = \bigotimes_{j=1}^n U_j \ket{\phi_j} \), where \( \ket{\phi} = \bigotimes_{j=1}^n \ket{\phi_j} \) is any fixed product eigenstate of a Pauli string Hamiltonian, and each \( U_j \sim \text{Haar}(\mathrm{SU}(2)) \) independently. Let \( \ket{b} = \bigotimes_{j=1}^n \ket{b_j} \) be a fixed computational basis state (or more generally, a Pauli product basis).

For any fixed Pauli string \( P_a = \sigma_{a_1} \otimes \cdots \otimes \sigma_{a_n} \) and fixed Pauli product basis state \( \ket{b} = \bigotimes_{j=1}^n \ket{b_j} \), the overlap can be written as
\begin{equation}
\langle b | P_a | \psi_0 \rangle 
= \prod_{j=1}^n \langle b_j | \sigma_{a_j} U_j | \phi_j \rangle =: \prod_{j=1}^n f_j(U_j).
\end{equation}

Since \( \mathrm{SU}(2) \) is uncountable and with a continuous Haar measure, the probability of sampling a unitary \( U_j \) such that \( f_j(U_j) = 0 \) is zero, and because the zero set is negligible relative to the rest of the space:
\begin{equation}
\mathbb{P}_{U_j}[f_j(U_j) = 0] = 0.
\end{equation}

Furthermore, since the \( U_j \) are sampled independently, the total overlap vanishes only if at least one of the \( f_j(U_j) \) vanishes. Thus,
\begin{equation}
\mathbb{P} \left[ \langle b | P_a | \psi_0 \rangle = 0 \right] \leq \sum_{j=1}^n \mathbb{P}[f_j(U_j) = 0] = 0.
\end{equation}

The same argument applies to the overlap \( \langle b | \psi_0 \rangle \), which corresponds to the case \( P_a = I \). Since the set of Pauli strings \( P_a \) is finite (with cardinality \( 4^n \)), the union over all \( a \) of the probability-zero events where \( \langle b | P_a | \psi_0 \rangle = 0 \) still has total probability zero:
\begin{equation}
\mathbb{P} \left[ \exists a : \langle b | P_a | \psi_0 \rangle = 0 \right] = 0.
\end{equation}

Hence, both \( \langle b | \psi_0 \rangle \) and all \( \langle b | P_a | \psi_0 \rangle \) are zero with probability zero. Therefore, the first-order-in-time measurement signal
\begin{equation}
p(b,t) \approx |\langle b | \psi_0 \rangle|^2 - 2t\, \mathrm{Im} \left[ \sum_a \theta_a \langle b | \psi_0 \rangle^* \langle b | P_a | \psi_0 \rangle \right]
\end{equation}
contains non-zero contributions from every \( \theta_a \), except on a measure-zero subset of initial states.
\end{proof}

\begin{remark}
Although the initial state \( \ket{\psi_0} = \bigotimes_j U_j \ket{0} \) is constructed using only local unitaries, the result holds for \emph{all} Pauli-string Hamiltonian components \( P_a \), including those that act nontrivially on multiple qubits. This is because the sensitivity condition concerns global overlaps such as \( \langle b | P_a | \psi_0 \rangle \), which are nonzero even when \( P_a \) is nonlocal. 

For instance, in a two-qubit system with \( P_a = X \otimes Y \) and \( \ket{b} = \ket{00} \), the overlap 
\begin{equation}
\langle 00 | (X \otimes Y) | \psi_0 \rangle = \langle 0 | X U_1 | 0 \rangle \cdot \langle 0 | Y U_2 | 0 \rangle
\end{equation}
is the product of two independent random complex numbers, each nonzero with probability one. Hence, the global overlap is generically nonzero even in the case of \( P_a \) being nonlocal.

\end{remark}

\section{Fisher Information Diagonalization}
\label{lemma:isotropy_spreading}
\begin{proof}
We consider the classical Fisher information matrix associated with measurement outcome probabilities \( \{p_i^{(r)}(\theta)\} \), obtained from an ensemble of experiments. Let \( H(\theta) \) be a Hamiltonian parametrized by real parameters \( \theta = (\theta_1, \dots, \theta_d) \),  so that \( H(\boldsymbol{\theta}) = \sum_j \theta_j H_j \). In each experiment, the initial state \( |\psi_r\rangle \) is generated by applying independent single-qubit unitaries \( U_\ell^{(r)} \in \mathrm{SU}(2) \) to a fixed product reference state \( |0\rangle^{\otimes n} \), i.e.,
\(
|\psi\rangle = \bigotimes_{\ell=1}^n U_\ell |0\rangle, \quad \text{with each } U_\ell \sim \text{Haar on } \mathrm{SU}(2).
\)
This defines an ensemble of spread states via local Haar-random sampling.\\ 

Each state is then evolved under the Hamiltonian \( H(\theta) = \sum_j \theta_j H_j \) for a short time \( t > 0 \), yielding
\begin{equation}
|\psi_r(t)\rangle = e^{-i H(\theta) t} |\psi_r\rangle.
\end{equation}
Measurement is performed in a fixed Pauli product basis with projectors \( \{ \Pi_i \} \), and the outcome probabilities are given by
\begin{equation}
p_i^{(r)}(\theta) = \langle \psi_r(t) | \Pi_i | \psi_r(t) \rangle.
\end{equation}

For small evolution times \( t > 0 \), the measurement probabilities admit a first-order expansion:
\begin{equation}
p_i^{(r)}(\theta) = \langle \psi_r | \Pi_i | \psi_r \rangle + i t \langle \psi_r | [H(\theta), \Pi_i] | \psi_r \rangle + O(t^2),
\end{equation}
so that the derivatives with respect to the each of Hamiltonians \(H(\boldsymbol{\theta}) = \sum_j \theta_j H_j,\) coefficients can be written as:
\begin{equation}
\frac{\partial p_i^{(r)}}{\partial \theta_j} =i t \langle \psi_r | [H_j, \Pi_i] | \psi_r \rangle + O(t^2).
\end{equation}

The Fisher information matrix for the \( r \)-th member of the ensemble is:
\begin{equation}
[\mathcal{I}_r(\theta)]_{jk} = \sum_i \frac{1}{p_i^{(r)}(\theta)} \frac{\partial p_i^{(r)}}{\partial \theta_j} \frac{\partial p_i^{(r)}}{\partial \theta_k}.
\end{equation}

Substituting the linear expansion in \( t \), we obtain:
\begin{equation}
[\mathcal{I}_r(\theta)]_{jk} = t^2 \sum_i \frac{1}{p_i^{(r)}(\theta)} \langle \psi_r | [H_j, \Pi_i] | \psi_r \rangle \langle \psi_r | [H_k, \Pi_i] | \psi_r \rangle + O(t^3).
\end{equation}

We now average over the ensemble of Haar-random product states:
\begin{equation}
\mathbb{E}_\psi[\mathcal{I}_{jk}(\theta)] = t^2 \sum_i \mathbb{E}_\psi \left[ \frac{1}{p_i^{(r)}(\theta)} \langle \psi_r | [H_j, \Pi_i] | \psi_r \rangle \langle \psi_r | [H_k, \Pi_i] | \psi_r \rangle \right] + O(t^3).
\end{equation}

To understand the scaling of the remainder term, we note that for \(k\)-local Hamiltonians each term \(H_j\) acts on at most \(k\) qubits. Since \(\Pi_i\) is a Pauli-product projector, the commutator \([H_j,\Pi_i]\) acts only on their overlapping qubits and its Pauli expansion therefore contains at most \(O(4^k)\) non-zero terms, which depends only on \(k\) and not on the system size.

The higher-order corrections involve nested commutators
\(
[H,[H,\Pi_i]]=\sum_{a,b}[h_a,[h_b,\Pi_i]],
\)
where \(H=\sum_a h_a\) denotes the decomposition of the Hamiltonian into local interaction terms. Only pairs \((a,b)\) whose supports overlap with that of \(\Pi_i\) contribute, since Pauli operators on disjoint supports commute. Because the interaction order and connectivity remain bounded, the number of such contributions depends only on \(k\). Consequently
\(
\|[H,[H,\Pi_i]]\| = O(4^k)
\)
independently of the system size \(N\), so the higher-order remainder does not scale with the Hilbert-space dimension \(d=2^N\).

Each commutator \( [H_j, \Pi_i] \) is a Hermitian operator that can be expanded as a linear combination of Pauli strings:
\begin{equation}
[H_j, \Pi_i] = \sum_\alpha c_\alpha^{(j,i)} Q_\alpha,
\end{equation}
where \( Q_\alpha \in \mathcal{P}_n \) are \( n \)-qubit Pauli strings and \( c_\alpha^{(j,i)} \in \mathbb{R} \) are coefficients depending on the structure of \( H_j \) and \( \Pi_i \).

Thus, the product of expectation values appearing in the Fisher information becomes a sum over Pauli string contributions:
\begin{equation}
\langle \psi | [H_j, \Pi_i] | \psi \rangle \langle \psi | [H_k, \Pi_i] | \psi \rangle = \sum_{\alpha, \beta} c_\alpha^{(j,i)} c_\beta^{(k,i)} \langle \psi | Q_\alpha | \psi \rangle \langle \psi | Q_\beta | \psi \rangle.
\end{equation}

Due to the unitary invariance of the Haar measure on \( \mathrm{SU}(2) \) and its independent application across qubits, the ensemble of one-local Haar-random product states satisfies \( \mathbb{E}_\psi[\langle \psi | Q | \psi \rangle] = 0 \) for any non-identity Pauli string \( Q \), and \( \mathbb{E}_\psi[\langle \psi | Q_\alpha | \psi \rangle \cdot \langle \psi | Q_\beta | \psi \rangle] = 0 \) for \( Q_\alpha \neq Q_\beta \).

Hence it follows that for \( j \ne k \), the expectation of the product of commutator terms vanishes in the ensemble average, leading to:
\begin{equation}
\mathbb{E}_\psi[\mathcal{I}_{jk}(\theta)] = 0 \quad \text{for } j \ne k,
\end{equation}

For the diagonal entries \( \mathcal{I}_{jj} \), the Fisher information contains terms of the form
\begin{equation}
\mathbb{E}_\psi\left[\langle \psi | Q_\alpha | \psi \rangle^2\right],
\end{equation}
which correspond to the variances of Pauli string expectation values. Since \( Q_\alpha \) is nontrivial and \( |\psi\rangle \) is locally Haar-random on each qubit, these variances are strictly positive. As a result,
\begin{equation}
\mathbb{E}_\psi[\mathcal{I}_{jj}(\theta)] > 0.
\end{equation}

By the matrix law of large numbers, the empirical average over \( R \) independent members of the ensemble converges to the expectation,
\begin{equation}
\lim_{R \to \infty} \left( \frac{1}{R} \sum_{r=1}^R \mathcal{I}_r(\theta) \right) = \mathrm{diag}(c_1, \dots, c_d), \quad \text{with } c_j > 0.
\end{equation}
\end{proof}

\section{Cumulative Fisher Information Scaling with Total Experiment Time}

\label{app:cumulative_fisher_information_scaling}

\begin{proof}
We consider sampling times \(t_k = \Delta t\, k^\alpha\), with \(\alpha > -1\). The total experiment time is
\begin{equation}
T_{\rm tot} = \Delta t \sum_{k=1}^{m_t} k^\alpha.
\end{equation}

As \(m_t \to \infty\), this sum satisfies
\begin{equation}
\sum_{k=1}^{m_t} k^\alpha
= \int_1^{m_t} x^\alpha\,dx + O(m_t^\alpha)
= \frac{m_t^{\alpha+1}}{\alpha+1}\bigl(1+o(1)\bigr)
\quad(m_t\to\infty),\qquad \alpha>-1,
\end{equation}
so we may write
\begin{equation} \label{eq:sum_approx}
\sum_{k=1}^{m_t} k^\alpha = \frac{m_t^{\alpha+1}}{\alpha+1} + O(m_t^\alpha),
\end{equation}
as $m_t \rightarrow \infty$, and thus
\begin{equation}
T_{\rm tot} = \Delta t \cdot \frac{m_t^{\alpha+1}}{\alpha+1} + O(m_t^\alpha).
\end{equation}

Solving for \(m_t\), we find the leading-order behavior:
\begin{equation}
m_t \sim \left( \frac{(\alpha+1)\, T_{\rm tot}}{\Delta t} \right)^{1/(\alpha+1)}.
\end{equation}

Next, consider the case where the Fisher information from time \(t\) scales as \(I(t) = \Theta(t^{\gamma_0})\) for each time stamp \(t_k\). Then the total Fisher information becomes
\begin{equation}
I_{\rm tot} = \sum_{k=1}^{m_t} t_k^{\gamma_0}
= (\Delta t)^{\gamma_0} \sum_{k=1}^{m_t} k^{\alpha\gamma_0}.
\end{equation}

Using a similar expansion as in Equation~\ref{eq:sum_approx}, we approximate:
\begin{equation}
\sum_{k=1}^{m_t} k^{\alpha\gamma_0} = \frac{m_t^{\alpha\gamma_0 + 1}}{\alpha\gamma_0 + 1} + O(m_t^{\alpha\gamma_0}),
\end{equation}
so the total Fisher information becomes
\begin{equation}
I_{\rm tot} = \Delta t^{\gamma_0} \cdot \frac{m_t^{\alpha\gamma_0 + 1}}{\alpha\gamma_0 + 1} + O(m_t^{\alpha\gamma_0}).
\end{equation}

Substituting the expression for \(m_t\) in terms of \(T_{\rm tot}\), we obtain
\begin{equation}
I_{\rm tot}
= C\, T_{\rm tot}^{\,p} \left(1 + \epsilon(m_t)\right),
\end{equation}
where
\begin{equation}
p = \frac{\alpha\gamma_0 + 1}{\alpha + 1}, \qquad \epsilon(m_t) = O(m_t^{-1}).
\end{equation}

To absorb the finite-\(m_t\) correction into the exponent, we define an effective exponent:
\begin{equation}
p_{\rm eff}(m_t)
:= \frac{\ln\bigl(I_{\rm tot}/C\bigr)}{\ln T_{\rm tot}}
= p + \frac{\ln(1+\epsilon(m_t))}{\ln T_{\rm tot}}.
\end{equation}

Here \(p_{\rm eff}\) is written as a function of \(m_t\) because the asymptotic analysis is performed for a fixed sampling schedule \(t_k=\Delta t\,k^\alpha\), with fixed \(\alpha>-1\) and fixed \(\Delta t>0\). Once this schedule is chosen, the experiment is extended by increasing only the number of sampling times \(m_t\). The total experiment time \(T_{\rm tot}\) is then the induced cumulative resource associated with this schedule. Equivalently, one may regard \(p_{\rm eff}\) as a function of \(T_{\rm tot}\), since \(T_{\rm tot}\) and \(m_t\) are asymptotically related through Eq.~\eqref{eq:sum_approx}.

Using the expansion \(\ln(1 + \epsilon(m_t)) = \epsilon(m_t) + O(\epsilon(m_t)^2)\) and noting that \(\ln T_{\rm tot} = \Theta(\ln m_t)\), we conclude
\begin{equation}
p_{\rm eff}(m_t)
= \frac{\alpha\gamma_0 + 1}{\alpha + 1} + O(m_t^{-1}).
\end{equation}

This yields the observed scaling:
\begin{equation}
I_{\rm tot} \approx T_{\rm tot}^{\,p_{\rm eff}(m_t)},
\end{equation}
with the effective exponent incorporating finite-\(m_t\) corrections.
\end{proof}

\section{Hamiltonian Recovery Method}
\label{sec:recovery}

To be able to numerically validate our previous theoretical claims, we consider probe‐state trajectories to learn Hamiltonians from, similar to those in \cite{wilde2022scalably, heightman2024solving, dutkiewicz2023advantage, hu2025ansatz}. While these works also leverage maximum‐likelihood estimation to fit model parameters, they focus on optimizing variational parameters in a fixed ansatz for the Hamiltonian in the so-called white-box scenario \cite{heightman2024solving}. In contrast, here we make no such structural assumptions, and directly optimize the Hamiltonian’s matrix entries and optimize to maximize the likelihood of measurement data obtained from the true Hamiltonian. The full implementation is available on GitHub \cite{githubrepo}.


We generate a dataset by evolving the system under a true Hamiltonian \(H_{\rm true}\) as follows:
\begin{enumerate}
  \item Reference state: 
      Prepare a fixed reference state, which we chose to be $ \ket{\psi_0} = \ket{0}^{\otimes N}$ without loss of generality \(\ket{\psi_0}\) with \(\rho_0 = \ket{\psi_0}\bra{\psi_0}\).
  \item Initial State Spreading:  
        Apply locally Haar‐random rotations
        \begin{equation}
          U_{\rm spread}
          = \bigotimes_{i=1}^n R_z(\xi_i)\,R_y(\chi_i)\,R_z(\phi_i),
        \end{equation}
        where 
        \(\chi_i = \arccos(1 - 2u_i)\), \(u_i\sim U[0,1]\), and 
        \(\phi_i,\xi_i\sim U[0,2\pi]\),
        yielding \(\rho_{\rm spread}=U_{\rm spread}\,\rho_0\,U_{\rm spread}^\dagger\).
  \item Time evolution:  
        Evolve \(\rho_{\rm spread}\) with the true Hamiltonian  \(H_{true}\) for a sequence of times
        \begin{equation}
          t_k = \Delta t\,k^\alpha,\quad k=1,2,\dots,m_t,
        \end{equation}
        where \(\Delta t>0\) and \(\alpha>-1\).  At each \(t_k\),
        \begin{equation}
          \rho(t_k)
          = e^{-iH_{\rm true}t_k}\,\rho_{\rm spread}\,e^{iH_{\rm true}t_k}.
        \end{equation}
  \item Measurement:  
        At each \(t_k\), measure in a random product Pauli basis
        \(p_j \in\{X,Y,Z\}^n\) and record the bit‑string outcomes \(s\).
\end{enumerate}


We use the empirical bit‑string distribution \(P_{\rm data}(b)\) obtained under a chosen configuration, comprising the initial‑state ensemble (via state spreading), evolution‑time schedule \((m_t,\alpha)\), product Pauli measurement bases, and repetition count. We then recover the Hamiltonian by  likelihood estimation over its independent matrix entries, matching the distribution simulated under the candidate Hamiltonian to \(P_{\rm data}\). The Hamiltonian estimate is parametrized by constructing \(\hat H(\theta)\) from a complex lower‑triangular map \(A(\theta)\) such that

\begin{equation}
          \hat H_{ij}(\theta)=
          \begin{cases}
            A_{ij}(\theta), & i \ge j,\\
            {A_{ji}^*(\theta)}, & i < j,
          \end{cases}
\end{equation}
enforcing Hermiticity in one step. We then use an extended parameter embedding for expressivity as follows: Embed the \(n^2\) real degrees of freedom into a slightly higher‑dimensional vector \(\theta\), linked to \(A(\theta)\) via a fixed neural-network architecture to smooth the optimization landscape (see Appendix~\ref{app:extended_parameter_embedding} for a detailed description). Each  likelihood estimation iteration then proceeds as follows:
\begin{enumerate}
  \item Simulating the experiment under the current estimate \(\hat H(\theta)\): for each initial state in \(\{\ket{\psi_{r}}\}_{r=1}^R\), evolve for times \(t_k\) and measure in the corresponding product Pauli basis \(p_j\), repeating \(s=1,\dots,S\) shots.
  \item Collecting simulated bit‑string outcomes \(b_{rkjd}\) to form the model distribution \(P_{\hat H(\theta)}(b)\).
\end{enumerate}

Based on the dataset \(D\) of size \(\lvert D\rvert = R\,K\,m_t\,S\), with entries indexed by \((r, j, k, s)\) and outcomes \(b_{rjks}\), we define the negative log-likelihood loss:
\begin{equation} \label{eq:loss}
  \mathcal{L}_D(\theta)
  = -\frac{1}{R J m_t S}\sum_{r=1}^{R}\sum_{k=1}^{J}\sum_{k=1}^{m_t}\sum_{s=1}^{S}
    \log P\bigl(b_{rjks}\mid t_k,p_j,\psi^{(r)}_0,\theta\bigr),
\end{equation}
where,
\begin{equation} \label{eq:meas_prob}
  P\bigl(b_{rjks}\mid \psi^{(r)}_0, p_j, t_k,\theta\bigr)
  = \bigl|\langle b_{rjks}\mid e^{-i\hat H(\theta)\,t_k}\mid\psi^{(r)}_0\rangle\bigr|^2.
\end{equation}
Gradients of \(\mathcal{L}_D(\theta)\) computed via back-propagation through the embedding and lower‑triangular map to update \(\theta\).  Minimizing \(\mathcal{L}_D(\theta)\) yields the Hermitian matrix that best reproduces the observed measurement statistics without imposing any additional bias. Once \(\mathcal{L}_D(\theta)\) has converged, we terminate the optimization and compute the reconstruction error via,
\begin{equation}
\label{eq:evaluation}
\varepsilon = \frac{1}{n^2}\sum_{i,j} \bigl| H^{\mathrm{true}}_{ij} - \hat{H}(\theta)_{ij}\bigr|,
\end{equation}
where \(H^{\mathrm{true}}\) and \(\hat H(\theta)\) denote the true and recovered Hamiltonian matrices, respectively.

\section{Extended Parameter Embedding}
\label{app:extended_parameter_embedding}

To smooth the optimization landscape and mitigate spurious local minima, while keeping a fully agnostic representation of each Hamiltonian entry, we embed the \(n^2\) real parameters into a slightly higher‑dimensional space via a compact feedforward network.  Crucially, the network’s input is fixed, so it does not learn a mapping but solely serves as an “extended parameter embedding” that outputs the Hamiltonian matrix entries directly in the Pauli basis.  Because this embedding is implemented as a standard neural network, its parameters can be updated via back‑propagation, avoiding the extra computational overhead that would arise from explicitly managing a larger parameter vector.\\

In the following we detail the executed adjustments.
We feed the network a constant vector \(x\in\mathbb{R}^{n^2}\) with uniform entries
    \begin{equation}
      x_\ell = \left(\frac{c}{\mathrm{dim}}\right)^p,\quad \mathrm{dim}=n^2,
    \end{equation}
typically \(c=1\), \(p=0\).  Because \(x\) never changes, the network parameters \(\theta\) fully determine the output, making the network a pure embedding of \(\theta\).\\
A two‑hidden‑layer feedforward network with \(\tanh\) activations, the first hidden layer containing 200 nodes and the second one containing 400 nodes, maps the fixed \(x\) through modest‑width layers to a real‑valued output of dimension \(n(n+1)/2\).\\
The network’s output vector is interpreted as the entries of a complex lower‑triangular matrix \(A(\theta)\).  The full Hermitian estimate \(\hat H(\theta)\) is then constructed by
\begin{equation}
  \hat H_{ij}(\theta)=
  \begin{cases}
    A_{ij}(\theta), & i\ge j,\\
    \overline{A_{ji}(\theta)}, & i<j,
  \end{cases}
\end{equation}
enforcing self‑adjointness in one step.

\section{Fitted Error-Scaling Exponents Spread-State Ensemble Sizes}
\label{app:spread_scaling_data}

In this appendix, we provide the full set of fitted exponents \(\beta \pm \delta\beta\) for each Hamiltonian family and ensemble size used in Figure~\ref{fig:scaling_state_spreading}. The exponents were obtained by fitting the reconstruction error \(\varepsilon\) to a power-law decay \(\varepsilon \propto T_{\rm tot}^{-\beta}\), where \(T_{\rm tot}\) denotes the total experiment time and \(\beta\) characterizes the scaling behaviour with increasing spread-state ensemble size \(R\).

\begin{table}[h!]
  \centering
\begin{tabular}{@{}rcccc@{}}
\toprule
\# Spreadings \(R\) & \textbf{XYZ} & \textbf{XYZ2} & \textbf{XYZ3} & \textbf{XXZ} \\
\midrule
1     & \(0.019 \pm 0.020\) & \(0.015 \pm 0.020\) & \(0.018 \pm 0.020\) & \(0.010 \pm 0.020\) \\
2     & \(0.067 \pm 0.021\) & \(0.071 \pm 0.021\) & \(0.066 \pm 0.021\) & \(0.078 \pm 0.022\) \\
4     & \(0.136 \pm 0.020\) & \(0.136 \pm 0.020\) & \(0.136 \pm 0.021\) & \(0.141 \pm 0.021\) \\
8     & \(0.244 \pm 0.018\) & \(0.250 \pm 0.018\) & \(0.254 \pm 0.018\) & \(0.253 \pm 0.019\) \\
16    & \(0.463 \pm 0.018\) & \(0.463 \pm 0.017\) & \(0.471 \pm 0.019\) & \(0.453 \pm 0.018\) \\
32    & \(0.671 \pm 0.007\) & \(0.656 \pm 0.005\) & \(0.660 \pm 0.004\) & \(0.678 \pm 0.004\) \\
64    & \(0.681 \pm 0.003\) & \(0.660 \pm 0.002\) & \(0.645 \pm 0.003\) & \(0.687 \pm 0.004\) \\
128   & \(0.711 \pm 0.003\) & \(0.661 \pm 0.003\) & \(0.646 \pm 0.003\) & \(0.706 \pm 0.003\) \\
\bottomrule
\end{tabular}

  \vspace{1em}
  \caption{Fitted error-scaling exponents \(\beta \pm \delta\beta\)  for each Hamiltonian family as a function of the spread state ensemble size $R$.}
  \label{tab:spread_scaling_data}
\end{table}

\section{Fitted Error-Scaling Exponents for Scheduling Exponent Sweep}
\label{app:alpha_scaling_data}

This appendix provides the tabulated error-scaling exponents \(\beta \pm \delta\beta\) as a function of the measurement-time scheduling exponent \(\alpha\), used in Figure~\ref{fig:scaling_alphas}. These values were obtained by fitting the reconstruction error \(\varepsilon\) to a power law \(\varepsilon \propto T_{\rm tot}^{-\beta(\alpha)}\) for each Hamiltonian family, where \(T_{\rm tot}\) denotes total experiment time. 

\begin{table}[h!]
  \centering
\begin{tabular}{@{}ccccc@{}}
\toprule
\(\alpha\) & \textbf{XYZ} & \textbf{XYZ2} & \textbf{XYZ3} & \textbf{XXZ} \\
\midrule
0.1 & \(0.449 \pm 0.007\) & \(0.446 \pm 0.005\) & \(0.443 \pm 0.005\) & \(0.444 \pm 0.006\) \\
0.2 & \(0.476 \pm 0.009\) & \(0.472 \pm 0.006\) & \(0.478 \pm 0.006\) & \(0.480 \pm 0.009\) \\
0.3 & \(0.520 \pm 0.011\) & \(0.507 \pm 0.006\) & \(0.504 \pm 0.009\) & \(0.508 \pm 0.013\) \\
0.4 & \(0.536 \pm 0.012\) & \(0.538 \pm 0.007\) & \(0.536 \pm 0.010\) & \(0.532 \pm 0.014\) \\
0.5 & \(0.564 \pm 0.013\) & \(0.558 \pm 0.008\) & \(0.559 \pm 0.009\) & \(0.557 \pm 0.014\) \\
0.6 & \(0.584 \pm 0.012\) & \(0.577 \pm 0.009\) & \(0.578 \pm 0.010\) & \(0.581 \pm 0.009\) \\
0.7 & \(0.611 \pm 0.011\) & \(0.597 \pm 0.008\) & \(0.596 \pm 0.012\) & \(0.615 \pm 0.004\) \\
0.8 & \(0.624 \pm 0.010\) & \(0.613 \pm 0.009\) & \(0.617 \pm 0.009\) & \(0.633 \pm 0.005\) \\
0.9 & \(0.642 \pm 0.009\) & \(0.639 \pm 0.007\) & \(0.634 \pm 0.009\) & \(0.653 \pm 0.005\) \\
1.0 & \(0.658 \pm 0.009\) & \(0.648 \pm 0.008\) & \(0.657 \pm 0.006\) & \(0.672 \pm 0.006\) \\
\bottomrule
\end{tabular}

  \vspace{1em}
  \caption{Fitted exponents \(\beta \pm \delta\beta\) for each Hamiltonian family as a function of the scheduling exponent \(\alpha\).}
  \label{tab:alpha_scaling_data}
\end{table}

\section{Fisher-Trace Data }
\label{app:fisher_trace_data}

In this appendix, we provide the Fisher-trace values used in Fig.~\ref{fig:fisher_diagnostic}(a) for the representative XYZ Hamiltonian family. The table lists the total experiment time \(T_{\rm tot}\) and the corresponding trace of the classical Fisher information \(\mathrm{Tr}\,\mathcal I\) for different spread-state ensemble sizes \(R\). These values are obtained from the diagnostic experiment described in Sec.~\ref{sec:fisher_diagnostics_experiments}.

\begin{table}[h!]
  \centering
\begin{tabular}{@{}rcccccc@{}}
\toprule
\(m_t\) & \(T_{\rm tot}\) & \(R=1\) & \(R=2\) & \(R=4\) & \(R=8\) & \(R=16\) \\
\midrule
2  & \(0.03\) & \(0.7777\) & \(1.5526\) & \(3.1017\) & \(6.2063\) & \(12.3922\) \\
4  & \(0.10\) & \(4.6614\) & \(9.3026\) & \(18.5877\) & \(37.2086\) & \(74.2927\) \\
6  & \(0.21\) & \(14.1174\) & \(28.1567\) & \(56.2797\) & \(112.7355\) & \(225.0867\) \\
8  & \(0.36\) & \(31.5793\) & \(62.9476\) & \(125.8703\) & \(252.3312\) & \(503.7890\) \\
10 & \(0.55\) & \(59.4385\) & \(118.4369\) & \(236.9069\) & \(475.2232\) & \(948.8533\) \\
\bottomrule
\end{tabular}

  \vspace{1em}
  \caption{Trace of the Fisher information \(\mathrm{Tr}\,\mathcal I\) as a function of the total experiment time \(T_{\rm tot}\) for different spread-state ensemble sizes \(R\), for the representative XYZ Hamiltonian family used in Fig.~\ref{fig:fisher_diagnostic}(a).}
  \label{tab:fisher_trace_data}
\end{table}

\section{Fitted Fisher-Scaling Exponents for Spread-State Ensemble Sizes}
\label{app:fisher_exponent_spreadings}

In this appendix, we provide the fitted Fisher-scaling exponents \(p \pm \delta p\) for each Hamiltonian family and spread-state ensemble size used in Fig.~\ref{fig:fisher_diagnostic}(b). The exponents were obtained by fitting the Fisher trace to the power law
\[
\mathrm{Tr}\,\mathcal I(T_{\rm tot}) \propto T_{\rm tot}^{p},
\]
where \(T_{\rm tot}\) denotes the total experiment time.

\begin{table}[h!t]
  \centering
\begin{tabular}{@{}rcccc@{}}
\toprule
\# Spreadings \(R\) & \textbf{XYZ} & \textbf{XYZ2} & \textbf{XYZ3} & \textbf{XXZ} \\
\midrule
1  & \(1.49104 \pm 0.00088\) & \(1.48973 \pm 0.00063\) & \(1.48913 \pm 0.00054\) & \(1.48990 \pm 0.00072\) \\
2  & \(1.49036 \pm 0.00079\) & \(1.48950 \pm 0.00058\) & \(1.48853 \pm 0.00048\) & \(1.48936 \pm 0.00065\) \\
4  & \(1.49078 \pm 0.00087\) & \(1.48928 \pm 0.00059\) & \(1.48867 \pm 0.00053\) & \(1.48980 \pm 0.00069\) \\
8  & \(1.49164 \pm 0.00101\) & \(1.49017 \pm 0.00071\) & \(1.48981 \pm 0.00065\) & \(1.49073 \pm 0.00076\) \\
16 & \(1.49162 \pm 0.00102\) & \(1.48992 \pm 0.00069\) & \(1.48949 \pm 0.00062\) & \(1.49096 \pm 0.00083\) \\
\bottomrule
\end{tabular}

  \vspace{1em}
  \caption{Fitted Fisher-scaling exponents \(p \pm \delta p\) for each Hamiltonian family as a function of the spread-state ensemble size \(R\).}
  \label{tab:fisher_exponent_spreadings}
\end{table}

\section{Fitted Fisher-Scaling Exponents for Scheduling-Exponent Sweep}
\label{app:fisher_exponent_alpha}

This appendix provides the fitted Fisher-scaling exponents \(p \pm \delta p\) as a function of the measurement-time scheduling exponent \(\alpha\), used in Fig.~\ref{fig:fisher_alpha_diagnostic}. These values were obtained by fitting the Fisher trace to the power law
\[
\mathrm{Tr}\,\mathcal I(T_{\rm tot}) \propto T_{\rm tot}^{p}
\]
for each Hamiltonian family at fixed spread-state ensemble size \(R=1\).

\begin{table}[h!]
  \centering
\begin{tabular}{@{}ccccc@{}}
\toprule
\(\alpha\) & \textbf{XYZ} & \textbf{XYZ2} & \textbf{XYZ3} & \textbf{XXZ} \\
\midrule
0.2 & \(1.13603 \pm 0.00328\) & \(1.13602 \pm 0.00328\) & \(1.13602 \pm 0.00328\) & \(1.13597 \pm 0.00328\) \\
0.4 & \(1.25283 \pm 0.00437\) & \(1.25278 \pm 0.00436\) & \(1.25277 \pm 0.00435\) & \(1.25267 \pm 0.00435\) \\
0.6 & \(1.34961 \pm 0.00391\) & \(1.34945 \pm 0.00388\) & \(1.34936 \pm 0.00385\) & \(1.34929 \pm 0.00388\) \\
0.8 & \(1.42816 \pm 0.00260\) & \(1.42769 \pm 0.00250\) & \(1.42732 \pm 0.00241\) & \(1.42754 \pm 0.00252\) \\
1.0 & \(1.49104 \pm 0.00088\) & \(1.48973 \pm 0.00063\) & \(1.48913 \pm 0.00054\) & \(1.48990 \pm 0.00072\) \\
\bottomrule
\end{tabular}

  \vspace{1em}
  \caption{Fitted Fisher-scaling exponents \(p \pm \delta p\) for each Hamiltonian family as a function of the scheduling exponent \(\alpha\) at fixed \(R=1\).}
  \label{tab:fisher_exponent_alpha}
\end{table}

\section{Fisher Trace Values for System-Size Sweep}
\label{app:fishertrace_qubits}

This appendix lists the numerical values of the trace of the classical Fisher information 
$\mathrm{Tr}\,\mathcal I$ (see Sec.~\ref{sec:fisher_diagnostics_theory}) as a function of the total experiment time $T_{\rm tot}$ for different 
system sizes \(n\), corresponding to the data shown in 
Fig.~\ref{fig:fisher_qubit_scaling}. The simulations use the XYZ Hamiltonian family with 
spread-state ensemble size \(R=1\) and scheduling exponent \(\alpha=1\).

\begin{table}[h!]
\centering
\begin{tabular}{@{}cccccc@{}}
\toprule
\(T_{\rm tot}\) & \(n=2\) & \(n=3\) & \(n=4\) & \(n=5\) & \(n=6\) \\
\midrule
0.03 & 0.01563 & 0.03542 & 0.07661 & 0.15617 & 0.31586 \\
0.10 & 0.09384 & 0.21262 & 0.45966 & 0.93651 & 1.89319 \\
0.21 & 0.28483 & 0.64517 & 1.39375 & 2.83791 & 5.73367 \\
0.36 & 0.63870 & 1.44660 & 3.12218 & 6.35294 & 12.82715 \\
0.55 & 1.20521 & 2.73009 & 5.88624 & 11.96797 & 24.14731 \\
\bottomrule
\end{tabular}

\vspace{1em}
\caption{Trace of the classical Fisher information $\mathrm{Tr}\,\mathcal I$ as a function of the 
total experiment time \(T_{\rm tot}\) for system sizes \(n=2\)–\(6\).}
\label{tab:fishertrace_qubits}
\end{table}

\section{Fisher-Matrix Diagonalization vs.\ Spread-State Ensemble Size}
\label{app:fisher_diag_R}

This appendix lists the numerical values of the Fisher-matrix diagonalization
measure \(\eta_{\mathrm{diag}}\) (see Sec.~\ref{sec:fisher_diagnostics_diagon}) as a function of the number of spread states \(R\),
corresponding to Fig.~\ref{fig:fisher_diag_eta}. The simulations use the XYZ
Hamiltonian family with total evolution time \(T_{\rm tot}=0.01\) and scheduling
exponent \(\alpha=1\).

\begin{table}[h!]
\centering
\begin{tabular}{@{}cccccc@{}}
\toprule
\(R\) & \(n=2\) & \(n=3\) & \(n=4\) & \(n=5\) & \(n=6\) \\
\midrule
1   & 0.65962 & 0.26696 & 0.14765 & 0.11116 & 0.16107 \\
2   & 0.65982 & 0.36614 & 0.28583 & 0.15573 & 0.12904 \\
4   & 0.86834 & 0.53498 & 0.39284 & 0.25370 & 0.19335 \\
8   & 0.88247 & 0.68136 & 0.53638 & 0.34291 & 0.26285 \\
16  & 0.90653 & 0.76999 & 0.68937 & 0.52652 & 0.38034 \\
32  & 0.93651 & 0.84513 & 0.79807 & 0.67375 & 0.50871 \\
64  & 0.94746 & 0.90253 & 0.87202 & 0.78486 & 0.67232 \\
128 & 0.96216 & 0.91517 & 0.91472 & 0.86922 & 0.79469 \\
\bottomrule
\end{tabular}

\vspace{1em}
\caption{Diagonalization measure \(\eta_{\mathrm{diag}}\) of the Fisher information
matrix as a function of the spread-state ensemble size \(R\) for system sizes
\(n=2\)–\(6\). Values close to \(1\) indicate that the Fisher matrix is largely
diagonal.}
\label{tab:fisher_diag_R}
\end{table}

\end{document}